\documentclass[aps,twocolumn,prl,reqno,superscriptaddress,reprint,floatfix]{revtex4-2}
\usepackage{etex}
\usepackage{graphicx}
\usepackage{hyperref}
\usepackage{physics}
\usepackage{float}
\usepackage{amsmath}

\usepackage{amsthm}
\usepackage{dsfont}
\usepackage{enumitem}
\usepackage{algcompatible}
\usepackage{algorithm}
\usepackage{marginnote}
\usepackage{dsfont}
\usepackage{mathtools}

\setlist{noitemsep}

\usepackage [
  n,
  advantage,
  operators,
  sets,
  adversary,
  landau,
  probability, 
  notions,
  logic,
  ff,
  mm,
  primitives,
  events,
  complexity,
  asymptotics,
  keys
  ] {cryptocode}

\usepackage{tikz, graphics}
\usetikzlibrary{arrows.meta, bending, patterns}
\usepackage{tikzscale} 

\floatname{algorithm}{Protocol}

\newcommand{\CZ}{\mathsf{CZ}}
\newcommand{\Correct}{\mathsf{Correct}}
\newcommand{\Redo}{\mathsf{Redo}}
\newcommand{\Abort}{\mathsf{Abort}}
\newcommand{\ok}{\mathsf{Ok}}

\newcommand{\Id}{\mathbb{I}}

\newcommand{\pmax}{p_{\mathit{max}}}
\newcommand{\pmin}{p_{\mathit{min}}}

\newcommand{\fu}{\mathtt u}

\newcommand{\fv}{\mathtt v}
\newcommand{\fV}{\mathtt V}
\newcommand{\fw}{\mathtt w}
\newcommand{\fW}{\mathtt W}
\newcommand{\ft}{\mathtt t}
\newcommand{\fT}{\mathtt T}
\newcommand{\fo}{\mathtt o}
\newcommand{\fO}{\mathtt O}
\newcommand{\fS}{\mathtt S}
\newcommand{\fail}{\mathrm{fail}}
\newcommand{\1}{\mathds{1}}

\newtheorem{theorem}{Theorem}
\newtheorem{lemma}{Lemma}
\newtheorem{corollary}{Corollary}
\newtheorem{definition}{Definition}
 
\begin{document}

\date{\today}
\title{Verifying BQP Computations on Noisy Devices with Minimal Overhead}
\author{Dominik Leichtle}
\affiliation{Laboratoire d’Informatique de Paris 6, CNRS, Sorbonne Université, 4 Place Jussieu, 75005 Paris, France}
\author{Luka Music}
\affiliation{Laboratoire d’Informatique de Paris 6, CNRS, Sorbonne Université, 4 Place Jussieu, 75005 Paris, France}
\author{Elham Kashefi}
\affiliation{School of Informatics, University of Edinburgh, 10 Crichton Street, Edinburgh EH8 9AB, United Kingdom}
\affiliation{Laboratoire d’Informatique de Paris 6, CNRS, Sorbonne Université, 4 Place Jussieu, 75005 Paris, France}
\author{Harold Ollivier}
\affiliation{INRIA, 2 rue Simone Iff, 75012 Paris, France}
\affiliation{Laboratoire d’Informatique de Paris 6, CNRS, Sorbonne Université, 4 Place Jussieu, 75005 Paris, France}

\begin{abstract}
  With the development of delegated quantum computation, clients will want to ensure confidentiality of their data and algorithms, and the integrity of their computations.
  While protocols for blind and verifiable quantum computation exist, they suffer from high overheads and from over-sensitivity: When running on noisy devices, imperfections trigger the same detection mechanisms as malicious attacks, resulting in perpetually aborted computations.
  We introduce the first blind and verifiable protocol for delegating BQP computations to a powerful server with repetition as the only overhead.
  It is composable and statistically secure with exponentially-low bounds and can tolerate a constant amount of global noise.
\end{abstract}

\maketitle

\section{Introduction}
\label{sec:intro}
Remotely accessible quantum computing platforms free clients from the burden of maintaining complex physical devices in house.
Yet, when delegating computations, they want their data and algorithms to remain private, and that these computations are executed as specified.
Several methods have been devised to achieve this (e.g.~\cite{bfk, fk}, see~\cite{GKK19} for a review).
Nonetheless, a practical solution remains to be found as all known protocols are too sensitive to noise.
Indeed, they have been designed for perfect devices, thus aborting as soon as the smallest deviation is detected.
Unfortunately, replacing such machines by even slightly noisy ones would make the verification procedure abort constantly, mistaking plain imperfections for the signature of malicious behaviour.

For dealing with this over-sensitivity, previous research either gave up on blindness~\cite{GHK18}, imposed restrictions on the noise model~\cite{KD19}, switched to a setting with two non-communicating servers and classical clients~\cite{MF13a}, or introduced computational assumptions~\cite{Urm18}.
Yet, these protocols either only achieve inverse-polynomial security or obtain exponential security by requiring an additional fault-tolerant encoding of the computation on top of the one used to suppress device noise.

We tackle this problem for BQP computations -- i.e.\ the class of decision problems that quantum computers can solve efficiently -- by introducing a protocol that provides noise-robustness, verification, blindness and delegation.
The protocol repeats the client's computation framed in the Measurement-Based Quantum Computation (MBQC) model -- a natural choice for delegating computations -- several times in a blind fashion while interleaving these executions with test rounds which aim at detecting a dishonest behaviour of the server.
A final majority vote over the computation rounds mitigates possible errors, thus providing the desired robustness. 

Combined with blindness, this forces the server to attack at least a constant fraction of the rounds to corrupt the computation, hence increasing its chances of getting caught by the tests.
Information theoretic security is proven in the composable framework of Abstract Cryptography~\cite{MR11}, ensuring security is not jeopardised by sequential or simultaneous instantiations with other protocols.

Crucially, our protocol has \emph{no space overhead} for each round when compared to the insecure computation in the MBQC model: the only price to pay for exponential security and correctness is a \emph{polynomial number of repetitions} of computations similar to the unprotected one.
This lets the client use the full extent of the available hardware for its computational tasks, and any increase in the capabilities of the quantum devices can be used entirely to scale-up these computations.
These properties make it, to our knowledge, the first experimentally realisable solution for verification of BQP computations, thus going beyond experimental feasibility demonstrations of verifiable building blocks~\cite{BKBF12,BFKW13,GRBM16,MPBM16} and potentially serving as a blueprint for the development of future quantum network applications.

\section{Preliminaries}
\label{sec:mbqc}
\paragraph{BQP Computations.}
The complexity class BQP contains the decisions problems that can be solved with bounded error probability using a polynomial size quantum cricuit.
More formally, a language $L$ is in BQP if there is a family of polynomial size quantum circuits which decides the language with an error probability of at most $p$.
The chosen value for $p$ is arbitrary as long as it is fixed, and is usually taken to be $1/3$.
Hence, a BQP computation for $L$ will have output $F(x) = 1$ for $x\in L$ with probability at least $1-p$, while it will have output $F(x) = 0$ for $x\notin L$ with probability at least $1-p$.
In the following, for a given BQP computation, $p$ will be referred to as the \emph{inherent error probability} to distinguish it from errors due to external causes such as the use of noisy devices.

\paragraph{Measurement Based Quantum Computation.}
An MBQC algorithm  (also called \emph{measurement pattern}) consists of a graph $G = (V,E)$, two vertex sets $I$ and $O$ defining input and output vertices, a list of angles $\{\phi_v \}_{v \in V}$ with $\phi_v \in \Theta := \{k\pi/4\}_{0 \leq k \leq 7}$ and a flow. To run it, the Client instructs the Server to prepare the graph state $\ket{G}$: for each vertex in $V$, the Server creates a qubit in the state $\ket{+}$ and performs a $\CZ$ gate for each pair of qubits in $E$. The Client then asks the Server to measure each qubit of $V$ along the basis $\qty{\dyad{+_{\phi'_v}}, \dyad{-_{\phi'_v}}}$ in the order defined by the flow of the computation, with $\ket{+_\alpha} = (\ket 0 + e^{i\alpha}\ket 1)/\sqrt 2$. The corrected angle $\phi'_v$ is given by $\phi'_v = (-1)^{s_v^X}\phi_v + s_v^Z\pi$ for binary values of $s_v^X$ and $s_v^Z$ that depend only on the outcomes of previously measured qubits and the flow. More details about the flow and the update rules for the measurement angles can be found in~\cite{hein2004multiparty,DK2006}.

As shown in \cite{mbqc}, the MBQC model is equivalent to the circuit model, so that any BQP algorithm in the circuit model can be translated in the MBQC model with at most polynomial overhead.

\paragraph{Hiding the Computation.}
A computation can be easily hidden if, instead of the Server preparing each qubit, the Client (i) for all $v \in V$ sends $\ket{+_{\theta_v}}$ with $\theta_v$ chosen uniformly at random in $\Theta$, (ii) asks the Server to measure the qubits in the basis defined by the angle $\delta_v = \phi'_v + \theta_v + r_v\pi$ for $r_v$ a random bit, while keeping $\theta_v$ and $r_v$ hidden from the Server, and (iii) uses $s_v = b_v \oplus r_v$ where $b_v$ is the measurement outcome to compute $s_v^X$ and $s_v^Z$ defined above. The angle $\theta_v$ here acts as a One-Time Pad for $\phi'_v$ while $r_v$ does the same for the measurement outcomes. This idea was first formalised in the \emph{Universal Blind Quantum Computation} (UBQC) Protocol in~\cite{bfk}.

\paragraph{Verifiability Through Trap Insertion.}
Verifiable protocols allow the Client to check that its computation has been done correctly.
To do this, the Client enlarges the graph used for the computation to insert traps.
These traps are made from qubits randomly prepared in $\ket{+_{\theta}}$ states and disconnected from the sub-graph used for performing the desired computation with the help of \emph{dummy qubits} -- i.e.~randomly initialised qubits sent by the Client in states $\{\ket{0}, \ket{1}\}$.
The first verification protocol via trapification was introduced in~\cite{fk}.
It was further optimised into the \emph{Verifiable Blind Quantum Computation Protocol} (or VBQC) of~\cite{KW15,XTH20}, achieving a linear overhead.

\section{Noise-Robust Verifiable Protocol}
\label{sec:prot}
\begin{figure*}[t]

\begin{minipage}{\textwidth}
\includegraphics[width=0.8\textwidth]{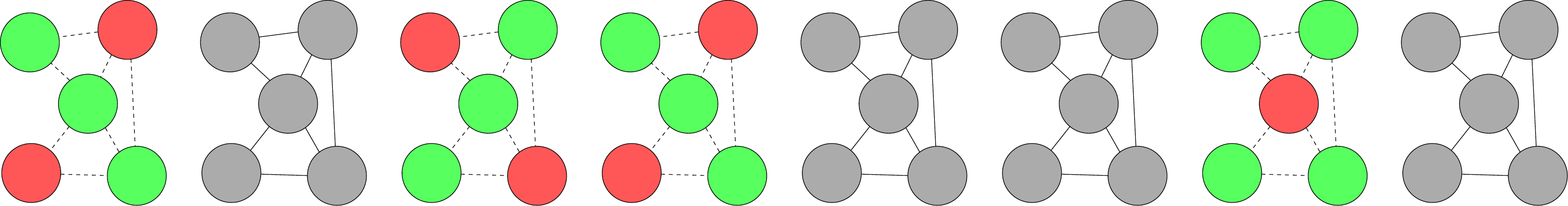}
	\caption{An example of rounds of the proposed protocol. Graphs in grey denote computation rounds while graphs containing red nodes (traps) and green nodes (dummies) are test rounds. Each qubit is always included in one type of test round. The Server remains completely oblivious of the differences between the rounds, which are solely known to the Client.}
  	\label{fig:protocol_highlevel}
\end{minipage}
  	
\begin{minipage}{\textwidth}

\begin{algorithm}[H]
\caption{\raggedright Noise-Robust VBDQC for BQP Computations}
\label{prot:MQ-VBQC}
\begin{algorithmic}[0]
\STATE \textbf{Client's Inputs:} Angles $\qty{\phi_v}_{v \in V}$ and flow $f$ on graph $G$, classical input to the computation $x \in \bin^{\#I}$ (where $\#X$ is the size of $X$).
\STATE \textbf{Protocol:}
\begin{enumerate}
\item The Client chooses uniformly at random a partition $(C, T)$ of $[n]$ ($C \cap T = \emptyset$) with $\#C = d$, the sets of indices of the computation and test rounds respectively.

\item[2.] For $j \in [n]$, the Client and the Server perform the following sub-protocol (the Client may send message $\Redo_j$ to the Server before step 2.c while the Server may send it to the Client at any time, both parties then restart round $j$ with fresh randomness):
\begin{enumerate}

\item[(a)] If $j \in T$ (test), the Client chooses uniformly at random a colour $\mathsf{V}_j \in_R \qty{V_k}_{k \in [K]}$ (this is the set of traps for this test round).
\item[(b)] The Client sends $\#V$ qubits to the Server. If $j \in T$ and the destination qubit $v \notin \mathsf{V}_j$ is a non-trap qubit (therefore a dummy), then the Client chooses uniformly at random $d_v \in_R \bin$ and sends the state $\ket{d_v}$. Otherwise, the Client chooses at random $\theta_v \in_R \Theta$ and sends the state $\ket{+_{\theta_v}}$.
\item[(c)] The Server performs a $\CZ$ gate between all its qubits corresponding to an edge in the set $E$.
\item[(d)] For $v \in V$, the Client sends a measurement angle $\delta_v$, the Server measures the appropriate corresponding qubit in the $\delta_v$-basis, returning outcome $b_v$ to the Client. The angle $\delta_v$ is defined as follows:
\begin{itemize}
\item If $j \in C$ (computation), it is the same as in UBQC, computed using the flow and the computation angles $\qty{\phi_v}_{v \in V}$. For $v \in I$ (input qubit) the Client uses $\tilde{\theta}_v = \theta_v + x_v\pi$ in the computation of $\delta_v$.
\item If $j \in T$ (test): if $v \notin \mathsf{V}_j$ (dummy qubit), the Client chooses it uniformly at random from $\Theta$; if $v \in \mathsf{V}_j$ (trap qubit), it chooses uniformly at random $r_v \in_R \bin$ and sets $\delta_v = \theta_v + r_v\pi$.
\end{itemize}
\end{enumerate}

\item[3.] For all $j \in T$ (test round) and $v \in \mathsf{V}_j$ (traps), the Client verifies that $b_v = r_v \oplus d_v$, where $d_v = \bigoplus_{i \in N_{G}(v)} d_i$ is the sum over the values of neighbouring dummies of qubit $v$. Let $c_{\mathit{fail}}$ be the number of failed test rounds (where at least one trap qubit does not satisfy the relation above), if $c_{\mathit{fail}} \geq w$ then the Client aborts by sending message $\Abort$ to the Server.

\item[4.] Otherwise, let $y_j$ for $j \in C$ be the classical output of computation round $j$ (after corrections from measurement results). The Client checks whether there exists some output value $y$ such that $\# \left\{ y_j \, | \, j \in C,\, y_j = y \right\} > \frac{d}{2}$. If such a value $y$ exists (this is then the majority output), it sets it as its output and sends message $\ok$ to the Server. Otherwise it sends message $\Abort$ to the Server.
\end{enumerate}
\end{algorithmic}
\end{algorithm}  	

\end{minipage}

\end{figure*}

Our Noise-Robust VBQC Protocol is formally defined in Protocol \ref{prot:MQ-VBQC} where test rounds are used in conjunction with computation rounds to provide verifiability.
We introduce it more intuitively in the next paragraphs and discuss the features that make it suitable for practical purposes.

\paragraph{Trap Insertion for BQP Computations.}
Because BQP computations have classical inputs and classical outputs, there exists a more economical trap insertion than what is available for quantum input and quantum output computations.
More concretely, it does not require any enlargement of the graph to insert traps alongside the computation.
Rather, the idea is to interleave pure \emph{computation rounds} (i.e.\ without inserted traps) and \emph{pure test rounds} (i.e.\ only made up of traps). 

Given a UBQC computation defined by a graph $G$, we construct test rounds based on a $k$-colouring $\{V_i\}_{i \in [k]}$ of $G$.
A partition of a graph in $k$ sets --~called colours~-- is a valid $k$-colouring if all adjacent vertices in the graph have different colours.
Therefore, by definition, a $k$-colouring satisfies $\bigcup_{i=1}^{k} V_i = V , \text{ and }\forall i \in [k],  \, \forall v \in V_i : N_G(v) \cap V_i = \emptyset$, where $N_G(v)$ are the neighbours of $v$ in $G$.
Hence, for each colour $i$, the Client can decide to insert traps for all vertices of $V_i$ and dummies in all other positions.
This defines the test round associated to colour $i$.
These tests require the same sequence of operations for the Server as regular UBQC computations, making them undetectable.

\paragraph{Informal Presentation of the Protocol.}
Suppose the Client wishes to delegate a BQP computation corresponding to a measurement pattern on a graph $G$ to the Server. The Client chooses a colouring $\{V_i\}_{i\in [k]}$ of $G$, and two integers $d$ and $t$. All these parameters are fixed for a given instantiation of the protocol and are publicly available to both parties.

The Client runs the UBQC Protocol $n := t+d$ times successively. For $d$ of the rounds chosen at random (computation rounds), the Client updates the measurement angles according to the measurement pattern of its desired computation. The remaining $t$ rounds are test rounds. For each such test round, the Client secretly chooses a colour at random and sends traps for vertices of that colour and dummies everywhere else. The Client instructs the Server to measure all qubits as in computation rounds, but with the measurement angle of trap qubits corresponding to the basis they were prepared in and a random measurement basis for the dummies. Because the trap qubits are isolated from each other, they should remain in their initial state. A test round is said to have \emph{passed} if all the traps yield the expected measurement results, and \emph{failed} otherwise. Figure~\ref{fig:protocol_highlevel} depicts such possible succession of rounds.

At the end of the protocol, the Client counts the number of failed test rounds. If this number is higher than a given threshold $w$, it aborts the protocol by sending the message $\Abort$ to the Server~\footnote{$w$ would typically be set by the Client given its \textit{a priori} understanding of the quality of the Server. As explained in the Discussion, this does not affect security: a higher value would induce more rounds than necessary to achieve a given confidence level, while a lower value would risk aborting with high probability.}. Otherwise it sets the majority outcome of the computation rounds as its output and sends message $\ok$ to the Server.

In this construction all rounds share the same underlying graph $G$, the same order for the measurements of qubits, and all angles are chosen from the same uniform distribution. We prove formally later that this implies blindness -- i.e.\ the Server cannot distinguish computation and test rounds, nor tell which qubits are traps -- which in turn makes this trap insertion strategy efficient to obtain verifiability. The parameters' range and influence on verifiability and noise-robustness bounds are detailed in the next section. 

\paragraph{Redo Feature.}
Because the Client or the Server may experience unintentional devices failures, they might wish to discard and redo a round $j \in [n]$. In this case, our protocol allows each party to send a $\Redo_j$ request to the other, in which case both parties simply repeat the exact same round albeit with fresh randomness. $\Redo_j$ requests are allowed only so long as the party asking for it is still supposed to be manipulating the qubits of round $j$. We show that this does not impact the blindness nor verifiability of the scheme. This means that a dishonest Server cannot use $\Redo$ requests to trick the Client into accepting an incorrect result. Such capability of our protocol is crucial in practice: without it, detected honest failures of devices happening during a test round would be counted as a failed test round, thus decreasing drastically the likelihood of successfully completing the protocol. Since concerned rounds can be safely repeated, the only consequence of experimental failures caught during an execution is an increase in the expected number of rounds.

\paragraph{Exponential Security Amplification.}
The above approach to trap insertion is efficient as the only overhead is the repetition of the same sub-protocol. Yet, using a single computation round and $n-1$ test rounds would leave at least $1/n$ chance for the Server to corrupt the computation. The only previously-known method to obtain an exponentially-low cheating probability was to insert traps into a single computation round at the expense of drastically increasing the graph's complexity and then using fault-tolerant encoding on top to amplify the security. By restricting the computation to BQP computations, we prove that a classical repetition error-correcting code is sufficient to achieve exponentially-low cheating probability.
This amplification technique is common in purely classical scenarios where attacks can be classically correlated across various rounds. Although this claim has been made as well in the quantum case in previous works \cite{fk,KW15,KD19}, it remained up to now unproven. The difficulty, which we address below, is that quantum attacks entangled across rounds are much more powerful than what classical correlations allow.  

\section{Security Results and Noise Robustness}
\label{sec:sec_rob}
This section presents the protocol's security properties in the Abstract Cryptography Framework of~\cite{MR11} (AC) and its noise-robustness on honest devices. The reader is referred to the Supplementary Material for formal definitions and proofs of Theorems~\ref{theo:sec_VDQC} and~\ref{thm:correctness}.

\paragraph{Security Analysis.}
\label{sec:corr_bl}
In AC, security is defined as indistinguishability between an Ideal Resource, which is secure by definition, and its real-world implementation, i.e.~the protocol.
This framework ensures a higher standard of security than in other approaches (see e.g.~\cite{KRBM07a} and Section 5.1 of~\cite{PR14}) and is inherently composable, meaning that security holds when the protocol is repeated sequentially or in parallel with others. 
This property is crucial as delegated protocols are important stepping stones towards more complex functionalities (e.g.\ subroutine for building Multi-Party Quantum Computation protocols \cite{KKMO21}).

Our security proof uses the results of~\cite{DFPR14} that reduce the composable security of a Verifiable Delegated Quantum Computation Protocol to four \emph{stand-alone criteria}:
\begin{itemize}
\item $\epsilon_{\mathit{cor}}$-local-correctness: the protocol with honest players produces the expected output;
\item $\epsilon_{\mathit{bl}}$-local-blindness: the Server's state at the end of the protocol is indistinguishable from the one which it could have generated on its own;
\item $\epsilon_{\mathit{ver}}$-local-verifiability: either the Client accepts a correct computation or aborts the protocol.
\item $\epsilon_{\mathit{ind}}$-independent-verification: the Server can determine on its own, using the transcript of the protocol and its internal registers, whether the Client will decide to abort or not. 
\end{itemize}
Then, the Local-Reduction Theorem (Corollary 6.9 from~\cite{DFPR14}) states that if a protocol implements a unitary transformation on classical inputs and is $\epsilon_{\mathit{cor}}$-locally-correct, $\epsilon_{\mathit{bl}}$-locally-blind and $\epsilon_{\mathit{ver}}$-locally-verifiable with $\epsilon_{\mathit{ind}}$-independent verification, then it is $\epsilon$-composably-secure with:
\begin{equation}\label{eq:reduction}
\epsilon = \mathit{max}\qty{\epsilon_{\mathit{sec}}, \epsilon_{\mathit{cor}}} \text{ and } \epsilon_{\mathit{sec}} := 4\sqrt{2\epsilon_{\mathit{ver}}} + 2\epsilon_{\mathit{bl}} + 2\epsilon_{\mathit{ind}}.
\end{equation}

With this at hand, we can state our main result:
\begin{theorem}[Security of Protocol \ref{prot:MQ-VBQC}]
\label{theo:sec_VDQC}
For $n = d+t$ such that $d/n$ and $t/n$ are fixed in $(0,1)$ and $w$ such that $w/t$ is fixed in $(0, \frac{1}{k}\cdot \frac{2p-1}{2p-2})$, where $p$ is the inherent error probability of the BQP computation, Protocol~\ref{prot:MQ-VBQC} with $d$ computation rounds, $t$ test rounds, and a maximum number of tolerated failed test rounds of $w$ is $\epsilon$-composably-secure with $\epsilon$
exponentially small in $n$.
\end{theorem}

\paragraph{Simple Upper-Bound on the Probability of Failure.}
The $\epsilon_{\mathit{ver}}$-local-verifiability amounts to upper bound the probability that an erroneous result is accepted by $\epsilon_{\mathit{ver}}$.
Given a BQP computation that decides whether $x$ belongs or not to the language $L$, our protocol would yield the correct result after the majority vote whenever less than $d/2$ computation rounds yield $F(x) \oplus 1$.
These erroneous results can be due to malicious behaviours of the server, to its use of noisy devices or to inherent errors of the BQP algorithm.
It is expected that, in $pd$ computation rounds, the BQP computation will give an inherently erroneous result, and that this will happen for a fraction greater than $p$ only with negligible probability.
Therefore, the result obtained by running our protocol will be correct whenever it is possible to guarantee that there is a negligible probability that the server corrupts more than $\left(\frac{1}{2} - p - \varphi \right)d$ computation runs for some $\varphi > 0$.
To this end, we use the trapification paradigm.
First, it ensures that each non-trivial deviation to the computation will be detected by at least one of the $k$ possible types of test rounds.
Second, because the deviations are distributed equally among test and computation runs, we can conclude that if less than $(\frac{1}{2} - p - \varphi - \varepsilon_1)t$ test runs are corrupted for some $\varepsilon_1 > 0$, then less than $(\frac{1}{2} - p - \varphi )d$ computations are corrupted with overwhelming probability.
This implies that setting $w = (\frac{1}{k} - \varepsilon_2)(\frac{1}{2} - p - \varphi - \varepsilon_1)t$ for $\varepsilon_2 > 0$ yields an exponentially low probability of failure.
Since $\varphi, \varepsilon_1, \varepsilon_2$ can be chosen arbitrarily small, we conclude that $\epsilon_{\mathit{ver}}$ can be made negligible for $0 < w/t < \frac{1}{k}(\frac{1}{2} -p)$.

\paragraph{Improved Upper-Bound on the Probability of Failure.}
The former bound can be improved by realising that some situations leading to incorrect results were double counted.
Indeed, we need to consider inherent errors from the BQP computation solely for the computation rounds that where unaffected by the Server's malicious behaviour.
This is due to the blindness of the scheme ensuring that the Server's deviation will be distributed equally among computation rounds with or without inherent errors.
Denoting by $m$ the total number of rounds affected by the Server's deviation, we expect $(md + (n-m)pd)/n$ computation rounds to be erroneous.
The first term comes from deviations of the Server, while the second comes from inherent errors in the BQP computation when the Server has not deviated on these rounds.
Requiring this quantity to be below $d/2$ amounts to guarantee that $m < \frac{2p-1}{2p-2} n$, which can be obtained following the line of arguments given in the previous paragraph whenever $w$ satisfies $0 < w/t < \frac{1}{k}\cdot \frac{2p-1}{2p-2}$.

\paragraph{Local-Correctness on Honest-but-Noisy Devices.}
\label{sec:noise_rob}
None of the stand-alone criteria introduced above consider device imperfections.
In fact, the analysis of correctness, blindness and verification makes no distinction between device imperfections and potentially malicious behaviours. 
Although satisfactory -- these properties make our protocol a concrete implementation of the Ideal Resource for Verifiable Delegated Quantum Computation --, it could still fall short of expectations in terms of usability because non malicious device imperfections could cause unintentional aborts.
Fortunately, for a class of realistic imperfections, our protocol is capable of correcting their impact and accepts with high probability.
In such case, the final outcome is the same as that obtained on noiseless devices with honest participants. 

This additional \emph{noise-robustness} property, the main innovation of this paper, means that Protocol~\ref{prot:MQ-VBQC} also satisfies the local-correctness property with negligible $\epsilon_{\mathit{cor}}$ for noisy but honest Client and/or Server.
This property holds under the following restrictions:
\begin{itemize}
\item The noise can be modelled by round-dependent Markovian processes -- i.e.\ a possibly different arbitrary CPTP map acting on each round.
\item The probability that at least one of the trap measurements fails in any single test round is upper-bounded by some constant $\pmax < \frac{1}{k}\cdot \frac{2p-1}{2p-2}$ and lower-bounded by $\pmin \leq \pmax$.
\end{itemize}

Theorem~\ref{thm:correctness} states that, in order for the protocol to terminate correctly with overwhelming probability on these noisy devices, $w$ should be chosen such that $w/t > \pmax$. Conversely, for any choice of $w/t < \pmin$, we show that the protocol aborts with overwhelming probability. 

\begin{theorem}[Local-Correctness of VDQC Protocol on Noisy Devices, Informal]
\label{thm:correctness}
As before, $p$ denotes the inherent error probability for the BQP computation. Assume a Markovian round-dependent model for the noise on Client and Server devices and let $\pmin \leq \pmax < \frac{1}{k}\cdot \frac{2p-1}{2p-2}$ be respectively a lower and an upper-bound on the probability that at least one of the trap measurement outcomes in a single test round is incorrect. If $w/t > \pmax$, Protocol~\ref{prot:MQ-VBQC} is $\epsilon_{\mathit{cor}}$-locally-correct with exponentially low $\epsilon_{\mathit{cor}}$. On the other hand, if $w/t < \pmin$, then the probability that Protocol~\ref{prot:MQ-VBQC} terminates without aborting is exponentially low.
\end{theorem}

Using again the Local-Reduction Theorem from~\cite{DFPR14}, this new bound concerning local-correctness on noisy devices can be combined with noise-independent blindness, input-independent verification and verifiability, to yield a composably secure protocol for $\epsilon = \mathit{max}\qty{\epsilon_{\mathit{sec}}, \epsilon_{\mathit{cor}}}$. Here, $\epsilon$ might depend on the noise level of the devices through $\epsilon_{\mathit{cor}}$. 
\section{Discussion}
\label{sec:concl}
\paragraph{Role of Noise Assumptions in Correctness Analysis.}\label{sec:noise_assump}
Our security proof does not rely on any assumption regarding the form or amplitude of the noise: it considers any deviation as potentially malicious and shows that the protocol provides information-theoretic verification and blindness.
The assumptions on the noise -- limited strength and markovianity -- are used only to show that correctness holds not only in the honest and noiseless case, but also when the imperfections of the devices are mild.
In such cases, their impact on the computation can be mitigated and the protocol will accept with high probability.

\paragraph{Fine-Tuning the Number of Repetitions.}
For specific computations with fixed security and correctness targets as well as noise levels, several parameters can be tuned to optimise the total runtime of our protocol.
First, distributing rounds across different machines is an effective way to reduce the overall execution time while composability ensures that security is preserved. 
Second, for a fixed graph, a smaller value of $k$ allows a larger value of $\pmax$, since exponential verification and correctness require $\pmax < w/t < \frac{1}{k}\cdot \frac{2p-1}{2p-2}$:
finding a small $k$-colouring of the graph used for the computation widens the gap between the chosen threshold ratio $w/t$ and $\frac{1}{k}\cdot \frac{2p-1}{2p-2}$, thereby reducing the number of rounds required to get the desired security and correctness levels.\footnote{This can be done once by the Server for its architecture and later shared with the Client before starting the protocol as a service.}
Third, the ratio $d/t$ also influences the number of repetitions.
Given fixed values for $p$, $k$, $w/t$, security and correctness levels, the optimal ratio can be determined numerically using equations~\ref{eq:bound}~and~\ref{eq:accept}, which explicitly relate the failure and success probabilities to these parameters.

\paragraph{Decoupling Verifiability and Fault-Tolerance.}
Because a single trap has bounded sensitivity -- the probability $\alpha$ of not detecting an attack at a given vertex is bounded away from $0$ -- it must be boosted to get exponential security. Previous work resorted to fault-tolerant encoding of the computation path to ensure that $r$ errors can be corrected (see \cite{fk,KW15}). 
This forces attackers to corrupt at least $r$ locations to affect the computation, which decreases the probability of not detecting such attacks to $\alpha^r$. Increasing the security of these protocols simultaneously increases the minimum distance of the fault-tolerant amplification scheme, thereby reducing the number of available qubits to perform the computation.

Our protocol's repetition of test rounds and majority vote serve the same purpose but with a much lighter impact.
Because our detection probability amplification relies on a classical procedure, all qubits can be devoted to useful computations irrespective of the desired security level.

Additionally, our protocol does not abort at the first failed trap while previous approaches do.
This means that, in the presence of noise, other protocols always require an exponentially low global residual error level to accept with overwhelming probability. 
On the contrary, our protocol only needs the average ratio of failed test rounds to be upper-bounded away from $\frac{1}{k}\cdot \frac{2p-1}{2p-2}$, which requires to bring the global residual error level to a constant only.
This promises to drastically ease experimental feasibility of verified quantum computations.

\paragraph{Acknowledgements.}
We thank Theodoros Kapourniotis and Atul Mantri for fruitful discussions.
We acknowledge support from the EU H2020 Program under grant agreement number 820445 (QIA). 
DL acknowledges support from the EU H2020 Program under grant agreement number ERC-669891 \mbox{(Almacrypt)}, and by the French ANR Projects ANR-18-CE39-0015 \mbox{(CryptiQ)} and ANR-18-CE47-0010 \mbox{(QUDATA)}.

\bibliographystyle{apsrev}
\bibliography{q_trinity}

\begin{thebibliography}{26}
\expandafter\ifx\csname natexlab\endcsname\relax\def\natexlab#1{#1}\fi
\expandafter\ifx\csname bibnamefont\endcsname\relax
  \def\bibnamefont#1{#1}\fi
\expandafter\ifx\csname bibfnamefont\endcsname\relax
  \def\bibfnamefont#1{#1}\fi
\expandafter\ifx\csname citenamefont\endcsname\relax
  \def\citenamefont#1{#1}\fi
\expandafter\ifx\csname url\endcsname\relax
  \def\url#1{\texttt{#1}}\fi
\expandafter\ifx\csname urlprefix\endcsname\relax\def\urlprefix{URL }\fi
\providecommand{\bibinfo}[2]{#2}
\providecommand{\eprint}[2][]{\url{#2}}

\bibitem[{\citenamefont{Broadbent et~al.}(2010)\citenamefont{Broadbent,
  Fitzsimons, and Kashefi}}]{bfk}
\bibinfo{author}{\bibfnamefont{A.}~\bibnamefont{Broadbent}},
  \bibinfo{author}{\bibfnamefont{J.}~\bibnamefont{Fitzsimons}},
  \bibnamefont{and} \bibinfo{author}{\bibfnamefont{E.}~\bibnamefont{Kashefi}},
  \emph{\bibinfo{title}{Measurement-Based and Universal Blind Quantum
  Computation}} (\bibinfo{publisher}{Springer Berlin Heidelberg},
  \bibinfo{address}{Berlin, Heidelberg}, \bibinfo{year}{2010}), pp.
  \bibinfo{pages}{43--86}, ISBN \bibinfo{isbn}{978-3-642-13678-8},
  \urlprefix\url{https://doi.org/10.1007/978-3-642-13678-8_2}.

\bibitem[{\citenamefont{Fitzsimons and Kashefi}(2017)}]{fk}
\bibinfo{author}{\bibfnamefont{J.~F.} \bibnamefont{Fitzsimons}}
  \bibnamefont{and} \bibinfo{author}{\bibfnamefont{E.}~\bibnamefont{Kashefi}},
  \bibinfo{journal}{Phys. Rev. A} \textbf{\bibinfo{volume}{96}},
  \bibinfo{pages}{012303} (\bibinfo{year}{2017}),
  \urlprefix\url{https://link.aps.org/doi/10.1103/PhysRevA.96.012303}.

\bibitem[{\citenamefont{Gheorghiu et~al.}(2019)\citenamefont{Gheorghiu,
  Kapourniotis, and Kashefi}}]{GKK19}
\bibinfo{author}{\bibfnamefont{A.}~\bibnamefont{Gheorghiu}},
  \bibinfo{author}{\bibfnamefont{T.}~\bibnamefont{Kapourniotis}},
  \bibnamefont{and} \bibinfo{author}{\bibfnamefont{E.}~\bibnamefont{Kashefi}},
  \bibinfo{journal}{Theory of Computing Systems} \textbf{\bibinfo{volume}{63}},
  \bibinfo{pages}{715} (\bibinfo{year}{2019}), ISSN \bibinfo{issn}{1433-0490},
  \urlprefix\url{https://doi.org/10.1007/s00224-018-9872-3}.

\bibitem[{\citenamefont{Gheorghiu et~al.}(2018)\citenamefont{Gheorghiu, Hoban,
  and Kashefi}}]{GHK18}
\bibinfo{author}{\bibfnamefont{A.}~\bibnamefont{Gheorghiu}},
  \bibinfo{author}{\bibfnamefont{M.~J.} \bibnamefont{Hoban}}, \bibnamefont{and}
  \bibinfo{author}{\bibfnamefont{E.}~\bibnamefont{Kashefi}},
  \bibinfo{journal}{Quantum Science and Technology}
  \textbf{\bibinfo{volume}{4}}, \bibinfo{pages}{015009} (\bibinfo{year}{2018}),
  \urlprefix\url{https://doi.org/10.1088/2058-9565/aaeeb3}.

\bibitem[{\citenamefont{Kapourniotis and Datta}(2019)}]{KD19}
\bibinfo{author}{\bibfnamefont{T.}~\bibnamefont{Kapourniotis}}
  \bibnamefont{and} \bibinfo{author}{\bibfnamefont{A.}~\bibnamefont{Datta}},
  \bibinfo{journal}{{Quantum}} \textbf{\bibinfo{volume}{3}},
  \bibinfo{pages}{164} (\bibinfo{year}{2019}), ISSN \bibinfo{issn}{2521-327X},
  \urlprefix\url{https://doi.org/10.22331/q-2019-07-12-164}.

\bibitem[{\citenamefont{Morimae and Fujii}(2013)}]{MF13a}
\bibinfo{author}{\bibfnamefont{T.}~\bibnamefont{Morimae}} \bibnamefont{and}
  \bibinfo{author}{\bibfnamefont{K.}~\bibnamefont{Fujii}},
  \bibinfo{journal}{Phys. Rev. Lett.} \textbf{\bibinfo{volume}{111}},
  \bibinfo{pages}{020502} (\bibinfo{year}{2013}),
  \urlprefix\url{https://link.aps.org/doi/10.1103/PhysRevLett.111.020502}.

\bibitem[{\citenamefont{Mahadev}(2018)}]{Urm18}
\bibinfo{author}{\bibfnamefont{U.}~\bibnamefont{Mahadev}}, in
  \emph{\bibinfo{booktitle}{59th {IEEE} Annual Symposium on Foundations of
  Computer Science, {FOCS} 2018, Paris, France, October 7-9, 2018}}, edited by
  \bibinfo{editor}{\bibfnamefont{M.}~\bibnamefont{Thorup}}
  (\bibinfo{publisher}{{IEEE} Computer Society}, \bibinfo{year}{2018}), pp.
  \bibinfo{pages}{259--267},
  \urlprefix\url{https://doi.org/10.1109/FOCS.2018.00033}.

\bibitem[{\citenamefont{Maurer and Renner}(2011)}]{MR11}
\bibinfo{author}{\bibfnamefont{U.}~\bibnamefont{Maurer}} \bibnamefont{and}
  \bibinfo{author}{\bibfnamefont{R.}~\bibnamefont{Renner}}, in
  \emph{\bibinfo{booktitle}{Innovations in Computer Science}}
  (\bibinfo{organization}{Tsinghua University Press}, \bibinfo{year}{2011}),
  pp. \bibinfo{pages}{1 -- 21}, ISBN \bibinfo{isbn}{978-7-302-24517-9},
  \urlprefix\url{https://conference.iiis.tsinghua.edu.cn/ICS2011/content/papers/14.html}.

\bibitem[{\citenamefont{Barz et~al.}(2012)\citenamefont{Barz, Kashefi,
  Broadbent, Fitzsimons, Zeilinger, and Walther}}]{BKBF12}
\bibinfo{author}{\bibfnamefont{S.}~\bibnamefont{Barz}},
  \bibinfo{author}{\bibfnamefont{E.}~\bibnamefont{Kashefi}},
  \bibinfo{author}{\bibfnamefont{A.}~\bibnamefont{Broadbent}},
  \bibinfo{author}{\bibfnamefont{J.~F.} \bibnamefont{Fitzsimons}},
  \bibinfo{author}{\bibfnamefont{A.}~\bibnamefont{Zeilinger}},
  \bibnamefont{and} \bibinfo{author}{\bibfnamefont{P.}~\bibnamefont{Walther}},
  \bibinfo{journal}{Science} \textbf{\bibinfo{volume}{335}},
  \bibinfo{pages}{303} (\bibinfo{year}{2012}), ISSN \bibinfo{issn}{0036-8075},
  \eprint{https://science.sciencemag.org/content/335/6066/303.full.pdf},
  \urlprefix\url{https://science.sciencemag.org/content/335/6066/303}.

\bibitem[{\citenamefont{Barz et~al.}(2013)\citenamefont{Barz, Fitzsimons,
  Kashefi, and Walther}}]{BFKW13}
\bibinfo{author}{\bibfnamefont{S.}~\bibnamefont{Barz}},
  \bibinfo{author}{\bibfnamefont{J.~F.} \bibnamefont{Fitzsimons}},
  \bibinfo{author}{\bibfnamefont{E.}~\bibnamefont{Kashefi}}, \bibnamefont{and}
  \bibinfo{author}{\bibfnamefont{P.}~\bibnamefont{Walther}},
  \bibinfo{journal}{Nature Physics} \textbf{\bibinfo{volume}{9}},
  \bibinfo{pages}{727} (\bibinfo{year}{2013}), ISSN \bibinfo{issn}{1745-2481},
  \urlprefix\url{https://doi.org/10.1038/nphys2763}.

\bibitem[{\citenamefont{Greganti et~al.}(2016)\citenamefont{Greganti, Roehsner,
  Barz, Morimae, , and Walther}}]{GRBM16}
\bibinfo{author}{\bibfnamefont{C.}~\bibnamefont{Greganti}},
  \bibinfo{author}{\bibfnamefont{M.-C.} \bibnamefont{Roehsner}},
  \bibinfo{author}{\bibfnamefont{S.}~\bibnamefont{Barz}},
  \bibinfo{author}{\bibfnamefont{T.}~\bibnamefont{Morimae}}, ,
  \bibnamefont{and} \bibinfo{author}{\bibfnamefont{P.}~\bibnamefont{Walther}},
  \bibinfo{journal}{New Journal of Physics} \textbf{\bibinfo{volume}{18}},
  \bibinfo{pages}{250} (\bibinfo{year}{2016}),
  \urlprefix\url{https://iopscience.iop.org/article/10.1088/1367-2630/18/1/013020}.

\bibitem[{\citenamefont{McCutcheon et~al.}(2016)\citenamefont{McCutcheon,
  Pappa, Bell, McMillan, Chailloux, Lawson, Mafu, Markham, Diamanti, Kerenidis
  et~al.}}]{MPBM16}
\bibinfo{author}{\bibfnamefont{W.}~\bibnamefont{McCutcheon}},
  \bibinfo{author}{\bibfnamefont{A.}~\bibnamefont{Pappa}},
  \bibinfo{author}{\bibfnamefont{B.~A.} \bibnamefont{Bell}},
  \bibinfo{author}{\bibfnamefont{A.}~\bibnamefont{McMillan}},
  \bibinfo{author}{\bibfnamefont{A.}~\bibnamefont{Chailloux}},
  \bibinfo{author}{\bibfnamefont{T.}~\bibnamefont{Lawson}},
  \bibinfo{author}{\bibfnamefont{M.}~\bibnamefont{Mafu}},
  \bibinfo{author}{\bibfnamefont{D.}~\bibnamefont{Markham}},
  \bibinfo{author}{\bibfnamefont{E.}~\bibnamefont{Diamanti}},
  \bibinfo{author}{\bibfnamefont{I.}~\bibnamefont{Kerenidis}},
  \bibnamefont{et~al.}, \bibinfo{journal}{Nature Communications}
  \textbf{\bibinfo{volume}{7}}, \bibinfo{pages}{13251} (\bibinfo{year}{2016}),
  ISSN \bibinfo{issn}{2041-1723},
  \urlprefix\url{https://doi.org/10.1038/ncomms13251}.

\bibitem[{\citenamefont{Hein et~al.}(2004)\citenamefont{Hein, Eisert, and
  Briegel}}]{hein2004multiparty}
\bibinfo{author}{\bibfnamefont{M.}~\bibnamefont{Hein}},
  \bibinfo{author}{\bibfnamefont{J.}~\bibnamefont{Eisert}}, \bibnamefont{and}
  \bibinfo{author}{\bibfnamefont{H.~J.} \bibnamefont{Briegel}},
  \bibinfo{journal}{Phys. Rev. A} \textbf{\bibinfo{volume}{69}},
  \bibinfo{pages}{062311} (\bibinfo{year}{2004}),
  \urlprefix\url{https://link.aps.org/doi/10.1103/PhysRevA.69.062311}.

\bibitem[{\citenamefont{Danos and Kashefi}(2006)}]{DK2006}
\bibinfo{author}{\bibfnamefont{V.}~\bibnamefont{Danos}} \bibnamefont{and}
  \bibinfo{author}{\bibfnamefont{E.}~\bibnamefont{Kashefi}},
  \bibinfo{journal}{Phys. Rev. A} \textbf{\bibinfo{volume}{74}},
  \bibinfo{pages}{052310} (\bibinfo{year}{2006}),
  \urlprefix\url{http://link.aps.org/doi/10.1103/PhysRevA.74.052310}.

\bibitem[{\citenamefont{Danos et~al.}(2007)\citenamefont{Danos, Kashefi, and
  Panangaden}}]{mbqc}
\bibinfo{author}{\bibfnamefont{V.}~\bibnamefont{Danos}},
  \bibinfo{author}{\bibfnamefont{E.}~\bibnamefont{Kashefi}}, \bibnamefont{and}
  \bibinfo{author}{\bibfnamefont{P.}~\bibnamefont{Panangaden}},
  \bibinfo{journal}{J. ACM} \textbf{\bibinfo{volume}{54}}
  (\bibinfo{year}{2007}), ISSN \bibinfo{issn}{0004-5411},
  \urlprefix\url{http://doi.acm.org/10.1145/1219092.1219096}.

\bibitem[{\citenamefont{Kashefi and Wallden}(2017)}]{KW15}
\bibinfo{author}{\bibfnamefont{E.}~\bibnamefont{Kashefi}} \bibnamefont{and}
  \bibinfo{author}{\bibfnamefont{P.}~\bibnamefont{Wallden}},
  \bibinfo{journal}{Journal of Physics A: Mathematical and Theoretical;
  preprint arXiv:1510.07408}  (\bibinfo{year}{2017}),
  \urlprefix\url{http://iopscience.iop.org/10.1088/1751-8121/aa5dac}.

\bibitem[{\citenamefont{Xu et~al.}(2020)\citenamefont{Xu, Tan, and
  Huang}}]{XTH20}
\bibinfo{author}{\bibfnamefont{Q.}~\bibnamefont{Xu}},
  \bibinfo{author}{\bibfnamefont{X.}~\bibnamefont{Tan}}, \bibnamefont{and}
  \bibinfo{author}{\bibfnamefont{R.}~\bibnamefont{Huang}},
  \bibinfo{journal}{Entropy} \textbf{\bibinfo{volume}{22}}
  (\bibinfo{year}{2020}), ISSN \bibinfo{issn}{1099-4300},
  \urlprefix\url{https://www.mdpi.com/1099-4300/22/9/996}.

\bibitem[{Note1()}]{Note1}
Note1, \bibinfo{note}{$w$ would typically be set by the Client given its
  \protect \textit {a priori} understanding of the quality of the Server. As
  explained in the Discussion, this does not affect security: a higher value
  would induce more rounds than necessary to achieve a given confidence level,
  while a lower value would risk aborting with high probability.}

\bibitem[{\citenamefont{K\"onig et~al.}(2007)\citenamefont{K\"onig, Renner,
  Bariska, and Maurer}}]{KRBM07a}
\bibinfo{author}{\bibfnamefont{R.}~\bibnamefont{K\"onig}},
  \bibinfo{author}{\bibfnamefont{R.}~\bibnamefont{Renner}},
  \bibinfo{author}{\bibfnamefont{A.}~\bibnamefont{Bariska}}, \bibnamefont{and}
  \bibinfo{author}{\bibfnamefont{U.}~\bibnamefont{Maurer}},
  \bibinfo{journal}{Phys. Rev. Lett.} \textbf{\bibinfo{volume}{98}},
  \bibinfo{pages}{140502} (\bibinfo{year}{2007}),
  \urlprefix\url{https://link.aps.org/doi/10.1103/PhysRevLett.98.140502}.

\bibitem[{\citenamefont{{Portmann} and {Renner}}(2014)}]{PR14}
\bibinfo{author}{\bibfnamefont{C.}~\bibnamefont{{Portmann}}} \bibnamefont{and}
  \bibinfo{author}{\bibfnamefont{R.}~\bibnamefont{{Renner}}},
  \bibinfo{journal}{arXiv e-prints} \bibinfo{eid}{arXiv:1409.3525}
  (\bibinfo{year}{2014}), \eprint{1409.3525}.

\bibitem[{\citenamefont{Kapourniotis et~al.}(2021)\citenamefont{Kapourniotis,
  Kashefi, Music, and Ollivier}}]{KKMO21}
\bibinfo{author}{\bibfnamefont{T.}~\bibnamefont{Kapourniotis}},
  \bibinfo{author}{\bibfnamefont{E.}~\bibnamefont{Kashefi}},
  \bibinfo{author}{\bibfnamefont{L.}~\bibnamefont{Music}}, \bibnamefont{and}
  \bibinfo{author}{\bibfnamefont{H.}~\bibnamefont{Ollivier}},
  \emph{\bibinfo{title}{Delegating multi-party quantum computations vs.
  dishonest majority in two quantum rounds}} (\bibinfo{year}{2021}),
  \eprint{2102.12949}.

\bibitem[{\citenamefont{Dunjko et~al.}(2014)\citenamefont{Dunjko, Fitzsimons,
  Portmann, and Renner}}]{DFPR14}
\bibinfo{author}{\bibfnamefont{V.}~\bibnamefont{Dunjko}},
  \bibinfo{author}{\bibfnamefont{J.~F.} \bibnamefont{Fitzsimons}},
  \bibinfo{author}{\bibfnamefont{C.}~\bibnamefont{Portmann}}, \bibnamefont{and}
  \bibinfo{author}{\bibfnamefont{R.}~\bibnamefont{Renner}}, in
  \emph{\bibinfo{booktitle}{Advances in Cryptology -- ASIACRYPT 2014}}, edited
  by \bibinfo{editor}{\bibfnamefont{P.}~\bibnamefont{Sarkar}} \bibnamefont{and}
  \bibinfo{editor}{\bibfnamefont{T.}~\bibnamefont{Iwata}}
  (\bibinfo{publisher}{Springer Berlin Heidelberg}, \bibinfo{address}{Berlin,
  Heidelberg}, \bibinfo{year}{2014}), pp. \bibinfo{pages}{406--425}, ISBN
  \bibinfo{isbn}{978-3-662-45608-8}.

\bibitem[{Note2()}]{Note2}
Note2, \bibinfo{note}{this can be done once by the Server for its architecture
  and later shared with the Client before starting the protocol as a service.}

\bibitem[{\citenamefont{Feller}(1991)}]{proba}
\bibinfo{author}{\bibfnamefont{W.}~\bibnamefont{Feller}},
  \emph{\bibinfo{title}{An Introduction to Probability Theory and Its
  Applications}} (\bibinfo{publisher}{John Wiley \& Sons},
  \bibinfo{year}{1991}), ISBN \bibinfo{isbn}{978-0-471-25709-7},
  \urlprefix\url{https://www.wiley.com/en-us/An+Introduction+to+Probability+Theory+and+Its+Applications%2C+Volume+2%2C+2nd+Edition-p-9780471257097}.

\bibitem[{\citenamefont{Greene and Wellner}(2017)}]{Greene_2017}
\bibinfo{author}{\bibfnamefont{E.}~\bibnamefont{Greene}} \bibnamefont{and}
  \bibinfo{author}{\bibfnamefont{J.~A.} \bibnamefont{Wellner}},
  \bibinfo{journal}{Bernoulli} \textbf{\bibinfo{volume}{23}},
  \bibinfo{pages}{1911–1950} (\bibinfo{year}{2017}), ISSN
  \bibinfo{issn}{1350-7265},
  \urlprefix\url{http://dx.doi.org/10.3150/15-BEJ800}.

\bibitem[{\citenamefont{Serfling}(1974)}]{Serf74}
\bibinfo{author}{\bibfnamefont{R.~J.} \bibnamefont{Serfling}},
  \bibinfo{journal}{Ann. Statist.} \textbf{\bibinfo{volume}{2}},
  \bibinfo{pages}{39} (\bibinfo{year}{1974}),
  \urlprefix\url{https://doi.org/10.1214/aos/1176342611}.

\end{thebibliography}

\appendix

\section{Useful Inequalities from Probability Theory}
\label{app:tail_bounds}
The following definitions and lemmata are useful tools for our proof. We refer the reader to \cite{proba} for more in-depth definitions.

\begin{definition}[Hypergeometric distribution]
        Let $N, K, n \in \mathbb{N}$ with $0 \leq n,K \leq N$.  A random variable $X$ is said to follow the \emph{hypergeometric distribution}, denoted as $X\sim\operatorname{Hypergeometric}(N,K,n)$, if its probability mass function is described by
	\begin{align*}
		\Pr \left[ X = k \right] = \frac{\binom{K}{k}\binom{N-K}{n-k}}{\binom{N}{n}}.
	\end{align*}
	As one possible interpretation, $X$ describes the number of drawn marked items when drawing $n$ items from a set of size $N$ containing $K$ marked items, \emph{without replacement}.
\end{definition}

\begin{lemma}[Tail bound for the hypergeometric distribution]
	Let $X\sim\operatorname{Hypergeometric}(N,K,n)$ be a random variable and $0 < t < K/N$. It then holds that
	\begin{align*}
	  &\Pr \left[ X \leq \left( \frac{K}{N}-t \right) n \right] \leq \exp \left( -2t^2n \right).
	\end{align*}
\end{lemma}

\begin{corollary} \label{cor:lower_tail}
	Let $X\sim\operatorname{Hypergeometric}(N,K,n)$ be a random variable and $0 < \lambda < \frac{nK}{N}$. It then holds that
	\begin{align*}
		\Pr \left[ X \leq \lambda \right] \leq \exp \left( -2n \left( \frac{K}{N} - \frac{\lambda}{n} \right)^2 \right).
	\end{align*}
\end{corollary}

\begin{lemma}[Serfling's bound for the hypergeometric distribution \cite{Greene_2017,Serf74}]
	Let $X\sim\operatorname{Hypergeometric}(N,K,n)$ be a random variable and $\lambda > 0$. It then holds that
	\begin{align*}
		\Pr \left[ \sqrt{n} \left( \frac{X}{n} - \frac{N}{K} \right) \geq \lambda \right] \leq \exp \left( - \frac{2\lambda^2}{1-\frac{n-1}{N}} \right).
	\end{align*}
\end{lemma}

\begin{corollary} \label{cor:upper_tail}
	Let $X\sim\operatorname{Hypergeometric}(N,K,n)$ be a random variable and $\lambda > \frac{nK}{N}$. It then holds that
	\begin{align*}
		\Pr \left[ X \geq \lambda \right] \leq \exp \left( -2n \left( \frac{\lambda}{n} - \frac{K}{N} \right)^2 \right).
	\end{align*}
\end{corollary}

Note the symmetry of Corollary~\ref{cor:lower_tail} and Corollary~\ref{cor:upper_tail}.

\begin{lemma}[Hoeffding's inequality for the binomial distribution] \label{lemma:hoeffding_binomial}
	Let $X\sim\operatorname{Binomial}(n,p)$ be a random variable. For any $k \leq np$ it then holds that
	\begin{align*}
		\Pr \left[ X \leq k \right] \leq \exp \left( -2\frac{(np-k)^2}{n} \right).
	\end{align*}
	Similarly, for any $k \geq np$ it holds that
	\begin{align*}
		\Pr \left[ X \geq k \right] \leq \exp \left( -2\frac{(np-k)^2}{n} \right).
	\end{align*}
\end{lemma}
 
\section{Formal Security Definitions}
\label{app:form_def}
We model $N$-round two party protocols between players $A$ (the honest Client) and $B$ (the potentially dishonest Server) as a succession of $2N$-CPTP maps $\{\mathcal{E}_i\}_{i\in [1,N]}$ and $\{\mathcal{F}_j\}_{j\in[1,N]}$. The maps $\{\mathcal E_i\}_i$ act on $\mathcal A$, $A$'s register, and $\mathcal C$, a shared communication register between $A$ and $B$. Similarly, the maps $\mathcal \{F_j\}_j$ act on $\mathcal B$ and $\mathcal C$. Note that $\mathcal B$ and the maps $\{\mathcal F_j\}_j$ can be chosen arbitrarily by $B$ and thus, unless $B$ is specified to be behaving honestly, there is no guarantee that they are those implied by our protocol. Since we are only interested in protocols where $A$ is providing a classical input $x$, we will equivalently write the input as the corresponding computational basis state $\ket x$ used to initialize $\mathcal A$, whereas $\mathcal B$ and $\mathcal C$ are initialized in a fixed state $\ket 0$.

Below, we denote by $\Delta(\rho, \sigma) = \frac{1}{2} \|\rho-\sigma\|$, the distance on the set of density matrices induced by the trace norm $\| \rho \| = \Tr \sqrt{\rho^\dagger\rho}$. We first define $\mathcal S$ the ideal resource for verifiable delegated quantum computation and then the local-properties from \cite{DFPR14}.

\paragraph{Ideal Resource for Verifiable Delegated Quantum Computation.} 
The ideal resource $\mathcal S$ has interfaces for two parties, $A$ and $B$. The $A$-interface takes two inputs: a classical input string $x$ and the description of $\mathcal U$, the computation to perform. The $B$-interface is filtered by a bit $b$. When $b=0$, there is no further legitimate input from $B$, while for $b=1$, it is allowed to send a bit $c$ that determines the output of the computation available at $A$'s interface. When $b=0$ or $c=0$, the output at $A$'s interface is equal to $\mathcal{M}_{\mathit{Comp}} \circ \mathcal{U}(\ket x)$, where $\mathcal{M}_{\mathit{Comp}}$ is the computational basis measurement. This corresponds to a ``no cheating'' behaviour. When $c=1$, $B$ decided to cheat and $A$ receives the $\Abort$ message which can be given as a quantum state of $\mathcal A$ which is taken orthogonal to any other possible output state. At $B$'s interface, $\mathcal S$ outputs nothing for $b=0$ while for $b=1$, $B$ receives $l(\mathcal U, x)$, the permitted leakage. For generic MBQC computations, the permitted leakage is set to $G$, the graph used in the computation. When $G$ is a universal graph for MBQC computation, the permitted leakage reduces to an upper-bound on the size of the computation $\# \mathcal U$.

For this ideal resource, the blindness is an immediate consequence of the server receiving at most the permitted leak, while verifiability is a consequence of the computation being correct when the server is not cheating while being aborted otherwise.

\paragraph{$\epsilon_{\mathit{cor}}$-Local-Correctness.} 
Let $\mathcal{P}_{AB}$ be a two-party protocol as defined above with the honest CPTP maps for players A and B. We say that such a protocol implementing $\mathcal{U}$ is \emph{$\epsilon_{cor}$-locally-correct} if for all possible inputs $x$ for $A$ we have:

\begin{equation}
\label{eq:correct}
\Delta\left(\Tr_B \circ \mathcal{P}_{AB} (\ket x), \mathcal U (\ket x)\right) \leq \epsilon_{cor}
\end{equation}

\paragraph{$\epsilon_{\mathit{bl}}$-Local-Blindness.}
\label{proof:local_blindness}
Let $\mathcal{P}_{AB}$ be a two-party protocol as defined above, and where the maps $\{\mathcal E_i\}_i$ are the honest maps. We say that such protocol is \emph{$\epsilon_{\mathit{bl}}$-locally-blind} if, for each choice of $\{\mathcal F_i\}_i$ there exists a CPTP map $\mathcal F' : L(\mathcal B) \rightarrow L(\mathcal B)$ such that, for all inputs $x$ for $A$, we have:

\begin{equation}
\label{eq:blind}
\Delta\left(\Tr_A \circ \mathcal{P}_{AB} (\rho), \mathcal F' \circ \Tr_A(\ket x)\right) \leq \epsilon_{\mathit{bl}}
\end{equation}

\paragraph{$\epsilon_{\mathit{ind}}$-Independent Verification.} 
Let $\mathcal{P}_{AB}$ be a verifiable 2-party protocol as defined above, where the maps $\{\mathcal E_i\}_i$ are the honest maps. Let $\bar B$ be a qubit extending $B$'s register and initialized in $\ket 0$. Let $\mathcal{Q}_{A\bar{B}} : L(\mathcal{A}\otimes\bar{\mathcal{B}}) \rightarrow L(\mathcal{A}\otimes\bar{\mathcal{B}})$ be a CPTP map which, conditioned on $\mathcal{A}$ containing the state $\ket{\Abort}$, switches the state in $\bar{\mathcal{B}}$ from $\ket 0$ to $\ket 1$ and does nothing in the other cases. 

We say that such a protocol's verification procedure is \emph{$\epsilon_{\mathit{ind}}$-independent} from player A's input if there exists CPTP maps $\mathcal{F}'_i : L(\mathcal{C} \otimes \mathcal{B}\otimes\bar{\mathcal{B}}) \rightarrow L(\mathcal{C} \otimes \mathcal{B}\otimes\bar{\mathcal{B}})$ such that:

\begin{equation}
\label{eq:indep}
\Delta\left(\Tr_A \circ \mathcal{Q}_{A\bar{B}} \circ \mathcal{P}_{AB}(\rho), \Tr_A \circ \mathcal{P}'_{AB\bar{B}}(\rho)\right) \leq \epsilon_{\mathit{ind}}
\end{equation}

where

\[
\mathcal{P}'_{AB\bar{B}} := \mathcal{E}_1 \circ \mathcal{F}'_1 \circ \ldots \circ \mathcal{E}_n \circ \mathcal{F}'_n
\]

\paragraph{$\epsilon_{\mathit{ver}}$-Local-Verifiability.} 
Let $\mathcal{P}_{AB}$ be 2-party protocols as defined above where the maps for $A$ are the honest maps, while the maps $\{\mathcal F_j\}_j$ for $B$ are not necessarily corresponding to the ideal (honest) ones. Let $x$ be the input given by $A$ in the form of a computational state $\ket x$ and $\mathcal U$ the computation it wants to perform. The protocols $\mathcal P_{AB}$ are \emph{$\epsilon_{\mathit{ver}}$-locally-verifiable} for $A$ if for each choice of CPTP maps $\{\mathcal F_j\}_{j}$, there exists $p \in [0,1]$ such that we have:

\[
\Delta\Bigl(\tr_B \mathcal P_{AB}(\ket x), p \mathcal U(\ket x) + (1-p) \ketbra{\Abort}) \leq \epsilon_{\mathit{ver}}
\]
 
\section{Composable Security}
\label{app:comp_sec}
In the paragraphs below, we show that our protocol satisfies each of the stand-alone criteria before combining them to get composable security.

\paragraph{Perfect Local-Correctness.} 
On perfect (non-noisy) devices, local-correctness is implied by the correctness of the underlying UBQC Protocol. 
This is because all the completed computation rounds correspond to the same deterministic UBQC computation, and that on such devices, general UBQC Protocols have been proven to be perfectly correct \cite{bfk, DFPR14}.
Thus $\epsilon_{\mathit{cor}} = 0$. 

\paragraph{Perfect Local-Blindness.} 
In case the computation is accepted, each round looks exactly like a UBQC computation to the Server.
Therefore the blindness comes directly from the composability of the various UBQC rounds that make our protocol~\cite{DFPR14}.
In case the computation is aborted, we need to take into account the fact that a possibly malicious Server could deduce the position of a trap qubit.
That could be the case if it attacked a single position in the test rounds and got caught.
Yet, as the position of the traps is not correlated to the input nor to the computation itself, knowing it does not grant additional attack capabilities to the Server, and blindness is recovered again as a consequence of the blindness of UBQC. 
More detailed statements can be found in the next section, where it is also shown that $\Redo$ requests have no effect on the local-blindness of the scheme.

\paragraph{Perfect Local-Independent-Verification.} 
Because in our protocol, the Client shares with the Server whether the computation was a success or an abort, this is trivially verified.

\paragraph{Exponential Local-Verifiability.} 
Local-verifiabili\-ty is satisfied if any deviation by the possibly malicious Server yields a state that is $\epsilon_{\mathit{ver}}$-close to a mixture of the correct output and the $\Abort$ message.
Equivalently, the probability that the Server makes the Client accept an incorrect outcome is bounded by $\epsilon_{\mathit{ver}}$.
Let $d/n$, $t/n$ and $w/t$ be the ratios of test, computation and tolerated failed test rounds.
Our protocol's local-verifiability is given by Theorem~\ref{thm:verif} and proven subsequently.

\paragraph{Proof of Exponential Composable-Security.} 
Our protocol has perfect correctness (for noiseless devices), blindness and input-independent verification. In addition, it is $\epsilon_{\mathit{ver}}$-locally-verifiable with $\epsilon_{\mathit{ver}}$ exponentially small in $n$. Therefore, by the Local-Reduction Theorem, it is $\epsilon$-composably-secure with $\epsilon = \epsilon_{\mathit{sec}} = 4\sqrt{2\epsilon_{\mathit{ver}}}$ and $\epsilon$ exponentially small in $n$.
Note that because we used the Local-Reduction Theorem to obtain fully composable security, we incurred an additional square root on our verifiability bound given by Equation~\ref{eq:reduction} and needed to satisfy the additional independence property.
This is of course not required if the protocol is only used sequentially with other schemes, which will probably be the case in early quantum computations since the machines will not be able to handle multiple protocols at the same time. In this case, the stand-alone model would be sufficient since it provides sequential composition, but would fail if parallel composition is needed.
 
\section{Proof of Perfect Local-Blindness}
\label{app:proof_lb}
\begin{proof}
To prove that Equation~\ref{eq:blind} holds for $\epsilon_{\mathit{bl}} = 0$, first note that at the end of our protocol, the Client $A$ reveals to the Server $B$ whether the computation was accepted or  aborted.
Hence, each case can be analyzed separately.
Second, we show that the interrupted rounds that have triggered a $\Redo$ can be safely ignored.
Indeed, each one of them is the begining of an interrupted UBQC computation, and, because UBQC is composable and perfectly blind \cite{DFPR14}, no information can leak to the Server through the transmitted qubits.
In addition, our protocol restricts the honest party $A$ in its ability to emit $\Redo$ requests, so that no correlations are created between the index of the interrupted rounds and $\mathcal U$ or the secret random parameters used in the rounds (angle and measurement padding, and trap preparations).
As a consequence, from the point of view of $B$, the state of the interrupted rounds is completely independent of the state of the non-interrupted ones and does not contain information regarding the input, computation or secret parameters.
That is, its partial trace over $A$ can be generated by $B$ alone.

For the non-interrupted rounds, we can invoke the same kind independence argument between the computation rounds and the test rounds.
As a result blindness of our protocol stems from the blindness of the underlying computation rounds.
In case the full protocol is a success, we can rely on the composability of the perfect blindness of each UBQC computation round to have perfect local-blindness.
For an abort, we can consider a situation that is more advantageous for $B$ by supposing that alongside the $\Abort$ message sent by $A$, it also gives away the location of the trap qubits.
In this modified situation, the knowledge of the computation being aborted does not bring additional information to $B$ as it only reveals that one of the attacked position was a trap qubit, which $B$ now already knows.
Using our independence argument between trap location on the one hand and the inputs, computation and other secret parameters, we conclude that revealing the location of the trap qubits does not affect the blindness of the computation rounds.
Hence, using composability again and combining the abort and accept cases, we arrive at Equation~\ref{eq:blind} with $\epsilon_{\mathit{bl}} = 0$.

\end{proof}
 
\section{Proof of Verifiability}
\label{app:proof-ver}
\begin{theorem}[Local Verifiability of Protocol \ref{prot:MQ-VBQC}]
\label{thm:verif}
Let $0 < w/t < \frac{1}{k}\cdot \frac{2p-1}{2p-2}$ and $0 < d/n < 1$ be fixed ratios, for $k$ different test rounds and where $p$ is the inherent error probability of the BQP computation.
Then, Protocol~\ref{prot:MQ-VBQC} is $\epsilon_{\mathit{ver}}$-locally-verifiable for exponentially-low $\epsilon_{\mathit{ver}}$.
\end{theorem}

\begin{proof}
Proving verifiability of a computation amounts to upper-bounding the probability of yielding a wrong output while not aborting.
This could be the result of the inherent randomness of the BQP computation that gives the wrong outcome with probability $p$, or of the server deviating from the instructed computation.
In the following, although rounds are expected to be run sequentially, the proof will examine the state of the {\em combined computation}.
This state corresponds to the server having simulaneous unrestricted access to all quantum systems sent by the client and possibly operating on them as a whole irrespectively of the underlying rounds they belong to.
In particular, the server could decide to perform some action on a qubit given measurements in one or several of the underlying runs, or to entangle the various underlying runs together. 

Note that, because the parties can only ask for redoing a run independently of the input, of the computation, of used randomness and of the output of the computation itself (comprising the result of trap measurements), interrupted runs can be safely ignored in the verification analysis as the state corresponding to these runs is uncorrelated to that of the completed runs.

\paragraph{Output of the combined computation.}
First, consider the output density operator $B(\{\mathcal F_j\}_j, \nu)$ representing all the classical messages the Client $A$ receives during its interaction with the Server $B$, comprising the final message containing the encrypted measurement outcomes.
Below, the CPTP maps $\{\mathcal F_j\}_j$ represent the chosen deviation of $B$ on the combined computation.
By encoding the classical messages as quantum states in the computational basis, the output density operator satisfies:
\begin{align}
\label{eq:output}
B(\{\mathcal F_j\}_j,\nu) = \Tr_B \bigg\{
&  \sum_b \ketbra{b+c_r}{b} \mathcal F \mathcal{P} \times \nonumber \\
& \quad \left(\ketbra{0}_B \otimes \ketbra{\Psi^{\nu,b}}\right) \times \nonumber \\
& \quad \mathcal{P}^\dagger {\mathcal F}^\dagger  \ketbra{b}{b+c_r}\bigg\} 
\end{align}
where \(b\) is the list of measurement outcomes defining the computation branch;
\(\nu\) is a composite index relative to the secret parameters chosen by $A$, i.e.~the type of each underlying run, the padding of the measurement angles and measurements outcomes and the trap setup;
\(\ketbra{b+c_r}{b}\) ensures that only the part corresponding to the current computation branch is taken into account and removes the One-Time-Pad encryption on non-output and non-trap qubits while leaving output and trap qubits unaffected, i.e.~encrypted;
\(\ketbra{0}_B\) is some internal register for $B$ in a fixed initial state;
and \(\ket{\Psi^{\nu,b}}\) is the state of the qubits sent by $A$ to $B$ at the beginning of the protocol tensored with quantum states representing the measurement angles of the computation branch $b$.

To obtain this result, the line of proof of~\cite{fk} can be applied to the combined computation.
This works by noting that for a given computation branch $b$ and given random parameters $\nu$, all the measurement angles are fully determined.
Therefore, provided that the computation branch is $b$, the measurement angles can be included into the initial state.
This defines $\ket{\Psi^{\nu,b}}$.
Then, each $\mathcal F_j$ is decomposed into an honest part and a pure deviation.
All the deviations are commuted and collected into $\mathcal F$ applied after $\mathcal P$, the unitary part of honest protocol, is applied.
The projections onto $\ket b$ then ensures that, after the deviation induced by $B$, the perceived computation branch is $b$.
This, together with the decrypting of non-output non-trap qubits, gives Equation~\ref{eq:output}.

\paragraph{Probability of failure.}
Recall that a failure for the combined computation on input $x$ occurs when the result after decrypting the outputs and performing the majority vote differs from $F(x)$ while the computation is accepted.

For the combined computation to be accepted, no more than \(w\) test runs should have a trap qubit measurement outcome opposite to what was expected.
Let \(\fT\) denote the set of trap qubits which is determined by $T$, the set of test runs, and the type of each test run.
In absence of any deviation on the combined computation, their expected value is \(\ket{r_\fT} = \bigotimes_{\ft \in \fT}\ket{r_\ft}\) where \(r_{\fT} = (r_{\ft})_{\ft \in \fT}\) denotes the measurement outcome padding values restricted to trap qubits.
Therefore, the projector onto the states of the trap qubits yielding to an accepted combined computation can be written as \( Q_\perp = \sum_{\fw \in \fW} X_{\fT}^{\fw} \ketbra{r_{\fT}} X_{\fT}^{\fw}\) with \(X_{\fT}^{\fw} = \bigotimes_{\ft\in \fT} X_{\ft}^{\fw_{\ft}}\), and where \(\fW\) is the set of length \(|\fT|\) binary vectors \(\fw\) that have at least a one in no more than \(w\) underlying (test) runs.

Similarly, define by $\fO$ the set of output qubits.
The correct value for these output qubit is $\ket{F(x)_{\fO} + r_{\fO}}$.
Then, for  \(\fV\) the set of length \(|\fO|\) binary vectors \(\fv\) that have at least $d/2$ ones in the underlying (computation) runs, the operator $P_\perp = \sum_{\fv \in \fV} X_{\fO}^{\fv} \ketbra{F(x) + r_{\fO}} X_{\fO}^{\fv}$ with \(X_{\fO}^{\fv} = \bigotimes_{\fo\in \fO} X_{\fo}^{\fv_{\fo}}\) is the projector onto the subspace of states that yield an incorrect result for the whole computation.
This is because when each output has been decrypted by the Client -- the one-time-padding $r_{\fO}$ is removed -- the majority vote will output $F(x) + 1$ because more than half of the outputs are equal to $F(x) +1$. 

Combining these two projectors allows to write the probability of failure:
\begin{align*}
& \Pr[\fail] = \\
& \quad \sum_{\nu}\sum_{b,k,\sigma,\sigma'} \Pr[\nu] \Tr \Big \{ \left(P_\perp \otimes  Q_\perp  \right) \times \\
& \quad \left( \alpha_{k\sigma}\alpha^*_{k\sigma'} \ketbra{b+c_r}{b} \sigma \mathcal{P}\ketbra{\Psi^{\nu, b}} \mathcal{P}^\dagger \sigma' \ketbra{b}{b+c_r}\right) \Big\}  
\end{align*}
where \(\mathcal F\) has been decomposed into Kraus operators indexed by \(k\), that were in turn decomposed onto the Pauli basis through the coefficients \(\alpha_{k\sigma}\) and \(\alpha_{k\sigma'}\).
Consequently, \(\sigma\) and \(\sigma'\) are Pauli matrices.

Using the explicit expressions for $P_\perp$ and $Q_\perp$, the above formula can be simplified:
\begin{align*}
  \Pr[\fail] = & \sum_{\nu}\sum_{\fv \in \fV, \fw \in \fW}\sum_{b',k,\sigma,\sigma'} \Pr[\nu] \Big \{ \\
  & \quad \bra{F(x)_{\fO} +r_{\fO}} \otimes \bra{r_{\fT}} \otimes \bra{b'} ( X_{\fO}^{\fv} \otimes X_{\fT}^{\fw}) \times \\
  & \quad \left( \alpha_{k\sigma}\alpha^*_{k\sigma'} \mathcal{P}\ketbra{\Psi^{\nu, b}} \mathcal{P}^\dagger \sigma' \right) \\
  & \quad (X_{\fO}^{\fv} \otimes X_{\fT}^{\fw}) \ket{F(x)_{\fO} +r_{\fO}} \otimes \ket{r_{\fT}} \otimes \ket{b'} \Big \}
\end{align*}
where $b'$ is the binary vector obtained from $b$ by restricting it to non-output and non-trap qubits.
This was obtained using the circularity of the trace and the fact that $\sum_b \bra{F(x)_{\fO}+r_{\fO}} \otimes \bra{r_{\fT}} ( X_{\fO}^{\fv}\otimes X_{\fT}^{\fw} ) \ketbra{b+c_r}{b} = \sum_{b'} \bra{F(x)_{\fO}+r_{\fO}} \otimes \bra{r_{\fT}} \otimes \ketbra{b'+c_r}{b'} ( X_{\fO}^{\fv}\otimes X_{\fT}^{\fw} )$ since there is no decoding for output and trap qubits -- i.e. $c_r$ is 0.

\paragraph{Using blindness of the scheme.} 
At this point, standard proofs of verifiability sum over the secret parameters defining the encryption to twirl the deviation of the Server and trace out non-trap qubits.
Here, because it is necessary to assess the probability of having more than half of the output qubits yielding the wrong measurement output $F(x) +1$, the trace shall be taken on non-trap and non-output qubits only.

The design of the protocol yielding the combined computation ensures blindness.
This implies that the resulting state of any set of qubits after applying \(\mathcal P\) and taking the average over their possible random preparations parameters is a completely mixed state.
This can be applied in the above equation for the set of non-output and non-trap qubits.
For output and trap qubits, the inner products must be computed before taking the sum over their random preparation parameters \(\nu_{\fO}\) and \(\nu_{\fT}\) respectively.

This gives:
\begin{align*}
\Pr[\fail] = & 
\sum_{\nu_{\fO},\nu_{\fT}, \fu}\sum_{\fv \in \fV, \fw \in \fW}\sum_{b',k,\sigma, \sigma'} \Pr[\nu_{\fO},\nu_{\fT}] \alpha_{k\sigma}\alpha^*_{k\sigma'}  \times \bigg\{ \\
& \quad  \bra{F(x)_{\fO}+r_{\fO}}  \otimes \bra{r_{\fT}} \otimes \bra{b'} (X_{\fO}^{\fv} \otimes X_{\fT}^{\fw}) \times\\
& \quad \sigma  \left( \ketbra{s_{\fO} + r_{\fO}} \otimes \ketbra{r_{\fT}} \otimes \frac{\Id}{\Tr \Id} \right) \sigma' \times \\
& \quad  (X_{\fO}^{\fv} \otimes X_{\fT}^{\fw}) \ket{F(x)_{\fO}+r_{\fO}} \otimes \ket{r_{\fT}}\otimes \ket{b'} \bigg\}
\end{align*}
where $\ket{s_{\fo}}$ is the state of the output qubit $\fo \in \fO$ when no deviation is applied by the Server. 

In the above equation, the contribution of each qubit factorizes.
For \(l\notin \fO \cup \fT \), because the Pauli matrices are traceless save for the identity, the only non-vanishing terms are obtained for \(\sigma_l = \sigma'_l\), where subscript \(l\) is used to select the action of \(\sigma\) and \(\sigma'\) on qubit \(l\).
In such case, the corresponding multiplicative factor equals 1.
A direct calculation shows that, for an output qubit $\fo \in \fO$,
\begin{align*}
  \sum_{r_{\fo}} & \bra{F(x)_{\fo}+r_{\fo}} X_{\fo}^{\fv_{\fo}} \sigma_{\fo} \ket{s_{\fo} + r_{\fo}} \\
  & \quad \bra{s_{\fo} + r_{\fo}} \sigma'_{\fo} X_{\fo}^{\fv_{\fo}} \ket{F(x)_{\fo}+r_{\fo}} = 0
\end{align*}
for \(\sigma_{\fo} \neq \sigma'_{\fo}\)
Similarly, for a trap qubit $\ft \in \fT$, \(\sum_{r_{\ft}} \bra{r_{\ft}} X_{\ft}^{\fw_{\ft}} \sigma_{\ft} \ketbra{r_{\ft}} \sigma'_{\ft} X_{\ft}^{\fw_{\ft}} \ket{r_{\ft}}\) vanishes for \(\sigma_{\ft} \neq \sigma'_{\ft}\).
Combining these yields:
\begin{align*}
\Pr[\fail]
& = \sum_{\nu_{\fO},\nu_{\fT}} \sum_{\fv \in \fV, \fw \in \fW} \sum_{k,\sigma} \Pr[\nu_\fO, \nu_\fT] |\alpha_{k\sigma}|^2  \times \\
& \qquad \prod_{\fo \in \fO} |\bra{F(x)_\fo+r_\fo} X_\fo^{\fv_\fo} \sigma_\fo \ket{s_\fo+r_\fo}|^2 \times \\
& \qquad \prod_{\ft\in \fT} |\bra{r_\ft} X_\ft^{\fw_\ft} \sigma_\ft \ket{r_\ft}|^2 \\
& = \sum_{k}\sum_{\sigma} |\alpha_{k\sigma}|^2 f(\sigma) 
\end{align*}
with 
\begin{align}\label{eq:f}
f(\sigma) =  & \sum_{\nu_\fO,\nu_\fT} \sum_{\fv\in \fV, \fw\in \fW}\Pr[\nu_\fO, \nu_\fT] \times  \nonumber \\
& \quad \prod_{\fo\in \fO} |\bra{F(x)_\fo+r_\fo} X_\fo^{v_\fo} \sigma_\fo \ket{s_\fo+r_\fo}|^2 \times \nonumber \\
& \quad \prod_{\ft\in \fT} |\bra{r_\ft} X_\ft^{\fw_\ft} \sigma_\ft \ket{r_\ft}|^2 
\end{align}
In short, this proves that the overall deviation $\mathcal F$ has the same effect as a convex combination of Pauli deviations $\sigma$ each occuring with probability $\sum_k |\alpha_{k, \sigma}|^2$.

\paragraph{Implicit upper bound.}
Because, \(\sum_{k, \sigma} |\alpha_{k\sigma}|^2 = 1\), the worst case scenario for the bound in Equation~\ref{eq:f} is when \(\alpha_{k\sigma} = 1\) for \(\sigma\) such that \(f(\sigma)\) is maximum.
Hence, the probability of failure is upper-bounded as follows:
\begin{align*}
\Pr[\fail] \leq \max_{\sigma} f(\sigma)
\end{align*}

Protocol~\ref{prot:MQ-VBQC} defines trap and output qubit configuration \(\nu_\fO,\nu_\fT\) by (i) the set \(\fT\) of trap qubits, itself entirely determined by the position and kind of test runs within the sequence of runs, and (ii) the preparation parameters \(\theta_l\) and \(r_l\) of each trap and output qubits.
Each parameter of (i) and (ii) being chosen independently, the probability of a given configuration \(\nu_\fO, \nu_\fT\) can be decomposed into the probability \(\Pr[\fT]\) for a given configuration of trap locations multiplied by the probability of a given configuration for the prepared state of the trap and output qubits, \(\prod_{l\in \fO \cup \fT}  \sum_{\theta_l,r_l}\Pr[\theta_l, r_l]\).
Using this, one can rewrite \(f(\sigma)\):
\begin{align}
f(\sigma) = & \sum_{\fT} \sum_{\fv\in \fV, \fw\in \fW} \Pr[\fT] \times \nonumber \\
& \; \prod_{\fo\in \fO} \sum_{\theta_\fo, r_\fo} \Pr[\theta_\fo, r_\fo] |\bra{F(x)_\fo + r_\fo} X_\fo^{\fv_\fo} \sigma_\fo \ket{s_\fo + r_\fo}|^2 \times \nonumber\\
& \; \prod_{\ft\in \fT} \sum_{\theta_\ft,r_\ft}\Pr[\theta_\ft, r_\ft] |\bra{r_\ft} X_\ft^{\fw_\ft} \sigma_\ft \ket{r_\ft}|^2 \label{eq:failure_prob}
\end{align}

For \(\sigma\) a Pauli deviation, denote by \(\sigma_{|X}\) the binary vector indexed by qubit positions of the combined computation where ones mark qubit positions for which \(\sigma\) acts as \(X\) or \(Y\).
Abusing notation, in the following, \(\fO\) will denote the  binary vector over qubit positions $i$ of the combined computation where ones are positioned for qubits in \(O\) -- that is the vector $(\1_{i\in\fO})_i$ for $i$ a qubit location.
Similarly, \(\fT\) will also denote \((\1_{i\in \fT})_i\).

Using the fact that \(|\bra{r_\ft} X_\ft^{\fw_\ft} \sigma_\ft \ket{r_\ft}|^2\) is 1 for \(X_\ft^{\fw_\ft} \sigma_\ft \in \{I,Z\}\) and 0 otherwise, the product over the trap qubits can be writen as:
\begin{align*}
& \prod_{\ft\in \fT}\sum_{\theta_\ft,r_\ft}\Pr[\theta_\ft, r_\ft] |\bra{r_\ft} X_\ft^{\fw_\ft} \sigma_\ft \ket{r_\ft}|^2 \\
& \quad =
\begin{cases}
1 & \mbox{for } \fT.\sigma_{|X} = \fw \\
0 & \mbox{otherwise} 
\end{cases}
\end{align*}
where, for \(a\) and \(b\) binary vectors, \(a.b\) is the bit-wise binary product vector.

For output qubits, before attempting the same computation, it is important to point out a important dependency of the deviation for the output qubits.
Failing to take it into account would yield an overly optimistic bound.
This dependency is due to the fact that, contrarily to trap qubits where the perfect protocol performs the identity, the output qubits are the result of more complex computation.
More precisely, the guarantee given by the protocol at this stage is only blindness.
Following the definition of the blind computing ideal resource given in the Formal Security Definitions Appendix above -- Equation~\ref{eq:blind} -- the Server is able to choose a deviation $\mathcal E$ and have it applied to the unprotected input of the protocol $x$, while himself not getting either $x$ nor $\mathcal E(x)$.
While $\mathcal E$ has been reduced here to a convex sum of Pauli deviations applied after the perfect protocol, nothing prevents these Pauli deviations to incorporate a dependency on the input $x$ or on the unencrypted output of the perfect protocol.
In short, this means that the Server could craft a deviation in such a way that only outputs equal to $F(x)$ are flipped, leaving those yielding $F(x)+1$ unaffected.

Going forward with the computation of factors for output qubits in Equation~\ref{eq:failure_prob}, it is thus necessary to distinguish output qubits that belong to computation rounds where no non-trivial deviation take place, and those that don't.
Define $\fu$ to be the random binary vector of length $|\fO|$ such that $s_\fo = F(x)+\fu_\fo$.
For an output qubit that is part of a computation round without a non-trivial deviation,
\begin{align*}
  & \sum_{\theta_\fo, r_\fo} \Pr[\theta_\fo, r_\fo] |\bra{F(x)_\fo + r_\fo} X_\fo^{\fv_\fo} \sigma_\fo \ket{s_\fo + r_\fo}|^2 \\
  & \quad = \sum_{\theta_\fo, r_\fo} \Pr[\theta_\fo, r_\fo, \fu_\fo] \times \\
  & \qquad |\bra{F(x)_\fo + r_\fo} X_\fo^{\fv_\fo} \sigma_\fo X_\fo^{\fu_\fo}\ket{F(x) + r_\fo}|^2 \\
  & \quad = \begin{cases}
    \Pr[\fu_\fo] & \mbox{for } \sigma_{|X, \fo} + \fu_\fo  = \fv_\fo \\
    0 & \mbox{otherwise} 
  \end{cases}
\end{align*}
When the output qubit is part of a computation round with a non-trivial deviation, the dependency argument given above yields:
\begin{align*}
  & \sum_{\theta_\fo, r_\fo} \Pr[\theta_\fo, r_\fo] \times \\
  & \quad |\bra{F(x)_\fo + r_\fo} X_\fo^{\fv_\fo} \sigma_\fo X_\fo^{\fu_\fo}\ket{F(x) + r_\fo}|^2 \leq \Pr[\fu_\fo]
\end{align*}

Hence, for a fixed $\sigma$, a necessary condition on $\fu$ and $\fT$ for having a non zero contribution to $f(\sigma)$ is thus:
\begin{align*}
  wt(\fT.\sigma_{|X}) \leq w \mbox{ and } wt(\fu.\neg\fS) \geq d/2 - |\fS|
\end{align*}
where $wt(.)$ is the Hamming weight of a binary vector, $\fS$ is a length $|\fO|$ binary vector where the ones are located on output qubits where at least one non-trivial deviation was performed in the corresponding computation round, and $\neg\fS$ is the bitwise negation of $\fS$.

Combining the corresponding bounds and summarizing the necessary condition with $(\fT, \fu) \in \Upsilon_\sigma$, one obtains:
\begin{align*}
  f(\sigma) \leq \sum_{(\fT, \fu) \in \Upsilon_\sigma} \Pr[\fT, \fu].
\end{align*}
Otherwise said, to record a failure of the protocol, the number of incorrect trap rounds need to be below the threshold $w$, while the number of non-trivially attacked computation rounds need to be greater than $d/2$ reduced by the amount of incorrect outcomes on non-attacked rounds due to the inherent randomness of the algorithm.

\paragraph{Explicit upper bound.}
Now, assume that the maximum of the bound above is attained for some $\sigma$ that happens to non-trivially affect one of the round, say \(k\), on more than one qubit.
Consider $\sigma'$ with the sole difference to $\sigma$ that $\sigma'$ restricted to one of these two qubits is equal to the identity. 
Then, $\sigma'$ still affects the round $k$ non-trivially, which implies that all configurations \((\fT,\fu)\) in \(\Upsilon_\sigma\) are also in \(\Upsilon_{\sigma'}\).
Therefore 
\begin{align*}
  \Pr[\fail] \leq \max_m \max_{\sigma \in E_{m}} \sum_{(\fT,\fu)\in\Upsilon_\sigma} \Pr[\fT, \fu].
\end{align*}
where $E_{m}$ denotes the set of Pauli operators with $m$ single qubit non-trivial deviations all in distinct rounds.

Because the bound above depends on $\fu$ only through $wt(\fu . \neg\fS)$ and because for any such subset the random variable $wt(\fu . \neg\fS)$ is less than $B(wt(\neg\fS), p)$ in the usual stochastic order, one obtains:
\begin{align*}
  \Pr[\fail] \leq \max_m \max_{\sigma \in E_{m}} \sum_{(\fT,\fu)\in\Upsilon_\sigma} \Pr[\fT]\times\Pr[\tilde \fu = \fu],
\end{align*}
where $\tilde \fu$ is a random binary vector where each coordinate follows a Bernouilli law with probability $p$, and where $B(n,p)$ is the binomial distribution for $n$ draws and probability $p$.
Using the fact that the random choice of test runs is completely uniform, the right hand side is invariant under permutations of the test and computation runs.
It is thus possible to restrict the range of the maximum to the specific Pauli operators $\sigma_{m}$ with a deviation on a single qubit in each of the first $m$ runs:
\begin{align}
  \Pr[\fail] \leq \max_m \sum_{T\in\Upsilon_{\sigma_{m}}} \Pr[\fT].\label{eq:upper_bound}
\end{align}

\begin{figure*}[ht]
  \centering
  \begin{minipage}{\textwidth}
    \includegraphics[width=0.9\textwidth]{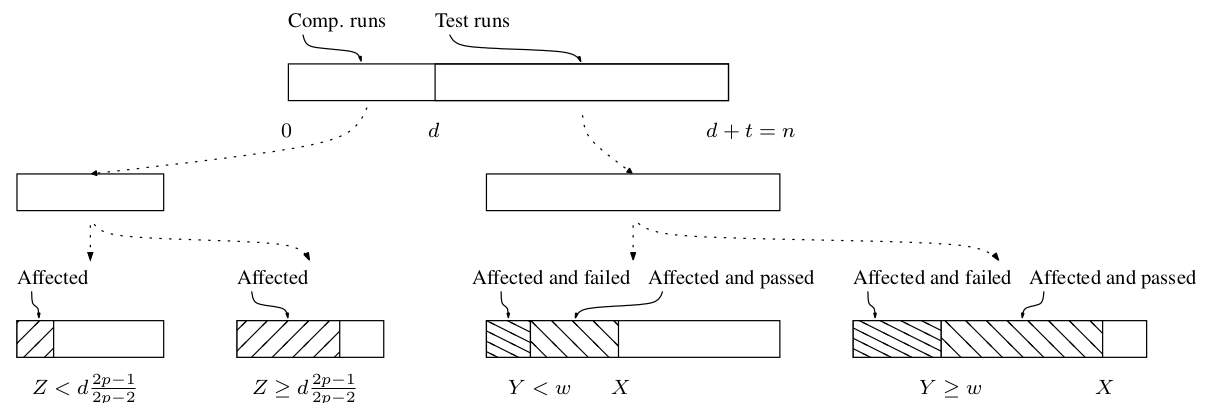}
    \caption{The four cases needed to determine a closed form upper bound for the probabiliity of failure. First, we determine the probability for the number of affected computation rounds. If it is low enough ($Z < d(2p-1)/(2p-2)$), no need to abort. If it is high ($Z \geq d(2p-1)/(2p-2)$), we find a bound on the probability that the number of failed test rounds $Y$ is below or above $w$.}
    \label{fig:vbqc}
  \end{minipage}
\end{figure*}

\paragraph{A closed from for the upper bound.}

To find a closed form upper bound for the soundness error, we now distinguish between two regimes for $m$, controlled by the parameter $\varphi > 0$:
\begin{enumerate}
	\item For $m \leq \left( \frac{2p-1}{2p-2} - \varphi \right) n$, we find a small upper bound on the probability that the client obtains a wrong result,
	\item whereas for $m \geq \left( \frac{2p-1}{2p-2} - \varphi \right) n$, we find a small upper bound on the probability that the client accepts the outcome of the protocol, i.e. that the verification passes.
\end{enumerate}
In the following, we define the constant ratios of test, computation and tolerated failed test runs as $\delta := d/n$, $\tau := t/n$ and $\omega := w/t$. Let $Z$ be a random variable counting the number affected computation runs (by the server's deviation or by inherent failure of the algorithm) and $Y$ a random variable counting the number of failed test runs, i.e.\ the number of affected test runs where the deviation hits a trap. We have that:
\begin{align*}
	\Pr \left[ \mathrm{fail} \right] & \leq \max_m \sum_{T\in\Upsilon_{\sigma_{m}}} \Pr [T] = \max_m \Pr \left[ Z \geq \frac{d}{2} \, \wedge \, Y \leq w  \right]\\
	& \leq \max \left\{ {\max_{m \leq \left( \frac{2p-1}{2p-2} - \varphi \right) n}} \Pr \left[ Z \geq \frac{d}{2} \right], \right. \\
	& \qquad\qquad\quad \left. {\max_{m \geq \left( \frac{2p-1}{2p-2} - \varphi \right) n}} \Pr \left[ Y \leq w \right] \right\}.
\end{align*}
Since $\Pr \left[ Z \geq d/2 \right]$ and $\Pr \left[ Y \leq w \right]$ are respectively increasing and decreasing with the number of attacked runs, both inner maximums are attained for $m = \left( \frac{2p-1}{2p-2} - \varphi \right)n$ and we therefore focus on this case.

Analogously to the verification proof of the original protocol, the second term can be bounded from above by first determining the minimum number of affected test runs before calculating the probability that the server's attack triggers a sufficient number of traps.

Hence, with $X$ denoting the number of test runs affected by the server's deviation, tail bounds for the hypergeometric distribution imply for all $\varepsilon_1 > 0$ that
\begin{align*}
	\Pr \left[ X \leq \left( \frac{m}{n} - \varepsilon_1 \right) t \right] \leq \exp \left( - \frac{2 \tau^2 \epsilon_1^2}{\frac{2p-1}{2p-2} - \varphi} n \right).
\end{align*}
Further, it follows by Hoeffding's bound for the binomial distribution that
\begin{align*}
	& \Pr \left[ \left. Y \leq \left( \frac{1}{k} - \varepsilon_2 \right) \left( \frac{m}{n} - \epsilon_1 \right) t \; \right| \; X = \left( \frac{m}{n} - \varepsilon_1 \right) t \right] \\
	& \quad \leq \exp \left( -2 \left( \frac{2p-1}{2p-2} - \varphi - \varepsilon_1 \right) \tau \varepsilon_2^2 n \right).
\end{align*}
All in all, we therefore obtain
\begin{align*}
	\Pr \left[ Y \leq w \right] \leq
	& \exp \left( - \frac{2 \tau^2 \epsilon_1^2}{\frac{2p-1}{2p-2} - \varphi} n \right) \\
	& + \exp \left( -2 \left( \frac{2p-1}{2p-2} - \varphi - \varepsilon_1 \right) \tau \varepsilon_2^2 n \right),
\end{align*}
where the threshold of tolerated failed test runs is set to $w = \left( 1/k - \epsilon_2 \right) \left( \frac{2p-1}{2p-2} - \varphi - \varepsilon_1 \right) t$.

Let's now focus on the first term and introduce the hypergeometrically distributed random variable $\bar{Z}$ counting the number of computation runs that are affected by the server's deviation. Then, for $\varepsilon_3 > 0$ tail bounds on the hypergeometric distribution imply
\begin{align*}
	\Pr \left[ \bar{Z} \geq \left( \frac{m}{n} + \varepsilon_3 \right) d \right] \leq \exp \left( - \frac{2 \delta^2 \varepsilon_3^2}{\frac{2p-1}{2p-2} - \varphi} n \right).
\end{align*}
Next, let $Z'$ be the random variable counting the number of computation runs that have not been affected by the server's deviation but which give a from $\bar{x}$ distinct result because of inherent failures of the algorithm. Note, that $Z'$ conditioned on $\bar{Z}$ fixed to a specific value is binomially distributed. It hence follows that
\begin{align*}
	& \Pr \left[ Z' \geq \left( p + \varepsilon_4 \right) \left( 1 - \frac{m}{n} - \varepsilon_3 \right) d \; \right| \; \left. \bar{Z} = \left( \frac{m}{n} + \varepsilon_3 \right) d \right] \\
	& \quad \leq \exp \left( -2 \left( 1 - \frac{2p-1}{2p-2} + \varphi - \varepsilon_3 \right) \delta \varepsilon_4^2 n \right).
\end{align*}
Note that it holds that $Z = \bar{Z} + Z'$. Therefore, it follows that
\begin{align*}
	& \Pr \left[ Z \geq \frac{d}{2} \right] \leq \Pr \left[ \left. Z \geq \frac{d}{2} \; \right| \; \bar{Z} \leq \left( \frac{m}{n} + \varepsilon_3 \right) d \right] \\
	& \qquad\qquad\qquad\quad + \Pr \left[ \bar{Z} \geq \left( \frac{m}{n} + \varepsilon_3 \right) d \right] \\
	& \quad \leq \Pr \left[ \left. Z' \geq \frac{d}{2} - \left( \frac{m}{n} + \varepsilon_3 \right) d \; \right| \; \bar{Z} = \left( \frac{m}{n} + \varepsilon_3 \right) d \right] \\
	& \qquad\quad + \Pr \left[ \bar{Z} \geq \left( \frac{m}{n} + \varepsilon_3 \right) d \right].
\end{align*}
Using the inequalities from above, we arrive at
\begin{align*}
	\Pr \left[ Z \geq \frac{d}{2} \right] \leq
	&\exp \left( -2 \left( 1 - \frac{2p-1}{2p-2} + \varphi - \varepsilon_3 \right) \delta \varepsilon_4^2 n \right) \\
	& + \exp \left( - \frac{2 \delta^2 \varepsilon_3^2}{\frac{2p-1}{2p-2} - \varphi} n \right)
\end{align*}
where we set
\begin{align*}
	\frac{d}{2} - \left( \frac{m}{n} + \varepsilon_3 \right) d = \left( p + \varepsilon_4 \right) \left( 1 - \frac{m}{n} - \varepsilon_3 \right) d.
\end{align*}
This condition can be rewritten as
\begin{align*}
	\frac{1}{2} - \frac{2p-1}{2p-2} + \varphi - \varepsilon_3 = \left( p + \varepsilon_4 \right) \left( 1 - \frac{2p-1}{2p-2} + \varphi - \varepsilon_3 \right),
\end{align*}
or equivalently
\begin{align*}
	\varepsilon_4 = & \left( 1 - \frac{2p-1}{2p-2} + \varphi - \varepsilon_3 \right)^{-1} \\
	& \qquad \cdot \left( \frac{1}{2} - \frac{2p-1}{2p-2} + \varphi - \varepsilon_3 \right) - p.
\end{align*}
It can be readily seen that this equation has solutions $\varepsilon_3, \varepsilon_4 > 0$ when $\varphi$ is fixed.

We finally conclude that
\begin{align}
	\nonumber \Pr \left[ \mathrm{fail} \right] \leq & \max \left\{ \exp \left( -2 \left( 1 - \frac{2p-1}{2p-2} + \varphi - \varepsilon_3 \right) \delta \varepsilon_4^2 n \right) \right. \\
	\nonumber & \qquad\quad + \exp \left( - \frac{2 \delta^2 \varepsilon_3^2}{\frac{2p-1}{2p-2} - \varphi} n \right), \\
	\label{eq:bound} & \qquad\quad \exp \left( - \frac{2 \tau^2 \epsilon_1^2}{\frac{2p-1}{2p-2} - \varphi} n \right) \\
	\nonumber & \qquad\quad + \left. \exp \left( -2 \left( \frac{2p-1}{2p-2} - \varphi - \varepsilon_1 \right) \tau \varepsilon_2^2 n \right) \right\}
\end{align}
for
\begin{align*}
	&w = \left( 1/k - \epsilon_2 \right) \left( \frac{2p-1}{2p-2} - \varphi - \varepsilon_1 \right) t, \\
	&0 < \varphi < \frac{2p-1}{2p-2}, \\
	&0 < \varepsilon_1 < \frac{1}{2} - \varphi, \\
	&0 < \varepsilon_2 < \frac{1}{k}, \\
	&0 < \varepsilon_3 < \varphi, \\
	& \varepsilon_4 = \left( 1 - \frac{2p-1}{2p-2} + \varphi - \varepsilon_3 \right)^{-1} \\
	& \qquad\qquad \cdot \left( \frac{1}{2} - \frac{2p-1}{2p-2} + \varphi - \varepsilon_3 \right) - p.
\end{align*}
To obtain an optimal bound, this expression must be minimized over $\varepsilon_1$, $\varepsilon_2$, $\varepsilon_3$ and $\varphi$.

Irrespective of the exact form of the optimal bound, choosing $\varphi$, $\varepsilon_1$, $\varepsilon_2$, and $\varepsilon_3$ sufficiently small implies the existence of protocols with verification exponential in $n$, for any fixed $0 < w/t < \frac{1}{k} \cdot \frac{2p-1}{2p-2}$ and fixed $\frac{d}{n}, \frac{t}{n} \in (0,1)$.

\paragraph{Optimality of the bound.}
To obtain the improved bound above, $Z_2$ was introduced as the count of non-affected computation runs yielding the correct result -- i.e.\ accept on yes instances, and reject on no instances.
Making sure that $Z_2$ would be greater than $d/2$ ensures that no matter what happens on computation runs that would yield an incorrect result, there is no possibility of being mistaken and reject in place of accept, and vice versa.
Yet, one might wonder if the situation is not more favorable: if the deviation by the server induces a flip of the accept / reject then could it be possible that some of the runs yielding incorrect result would be corrected by the deviation.
At first sight, this could be motivated by the fact that the computation being blind, the server could not possibly craft an attack that would selectively affect the runs yielding the correct results.
Unfortunately, this intuition is wrong: blindness does not rule out attacks that have different effects depending on the result of the computation itself.

To see this, consider the following situation.
Consider an algorithm solving a decision problem deterministically, so that in case of a yes instance, the algorithm outputs $\ket +$, and, in case of a no instance the output  is $\ket -$.
This deterministic algorithm yields a trivial randomized algorithm where a second qubit is generated in state $\alpha \ket 0 + \beta \ket 1$, with $|\alpha|^2 > 2/3$.
The new algorithm would take the output of the first one and apply a control-$Z$ gate between both qubits so that when the second qubit is traced out, the first one yields the correct answer with probability $|\alpha|^2$.
Yet, nothing could rule out an alternate implementation where after the control-$Z$ gate, the state of the first qubit undergoes two $H$ gates controlled by the second qubit being $\ket 0$.
Clearly this operation applies the identity to the first qubit as $H^2 = I$.
However, if the server applies a $X$ gate on the first qubit between these two control-$H$ gates, it will amount to a deviation consisting of a $Z$ gate applied only when the second qubit is $\ket 0$.
As a result, its attack only affects runs with the correct result.
Note that the attack affects correct outcomes only because in between the two control-$H$ gates, the computational branch for correct outcomes yields a state in the computational basis, while for incorrect ones it is the $\ket \pm$ basis.
This property is true independently of the quantum one-time-pad encryption of the states and can hence be applied on an encrypted computation. 

This example might seem excessively artificial, but such situations cannot be ruled out a priori, i.e.\ without an extensive understanding of the algorithm being implemented and of the proposed implementation.
In fact, a similar situation \cite{KKMO21} has already been encountered in the context of multi-party quantum computation where attacks could be crafted to evade detection when using less obvious inappropriate implementations.
\end{proof}
 
\section{Proof of Noise-Robustness}
\label{app:proof_rob}
Recall that the constant ratios of test, computation and tolerated failed test rounds are given by $\delta = d/n$, $\tau = t/n$ and $\omega = w/t$. We define the acceptance of the protocol to be the probability that the Client does not abort at the end of an execution. We then bound this probability in two regimes: (i) if the maximal noise $\pmax$ is smaller the (ratio) threshold of failed test runs, the protocol accepts with high probability; (ii) if the noise of the device is too large, i.e.\ $\pmin$ is already too large compared to the threshold, the protoco will most certainly abort.

\begin{lemma}[Acceptance on Noisy Devices]
\label{lem:acceptance}
Assume a Markovian round-dependent model for the noise on the Client and Server devices and let $\pmin \leq \pmax < 1/2$ be respectively a lower and an upper-bound on the probability that at least one of the trap measurement outcomes in a single test round is incorrect.

If $\omega > \pmax$, then the probability that the Client does not accept at the end of Protocol~\ref{prot:MQ-VBQC} is bounded by exponentially small $\epsilon_{\mathit{rej}}$ where
\begin{equation} \label{eq:accept}
	\epsilon_{\mathit{rej}} = \exp \left( -2 (\omega - \pmax)^2 \tau n \right).
\end{equation}
On the other hand, if $\omega < \pmin$, then the Client's acceptance in Protocol~\ref{prot:MQ-VBQC} is exponentially small and bounded by $\exp \left( -2 (\pmin - \omega)^2 \tau n \right)$.
\end{lemma}

\begin{proof}
We define the random variables $Y$ that corresponds to the number of failed test rounds during one execution of the protocol. We call $\ok$ the event that the Client accepts at the end of the protocol -- if not too many test rounds fail, meaning that $Y < w$.

\paragraph{For $\omega > \pmax$.} Equivalently, we have that $w > t\pmax$. We are looking to lower-bound the probability that an honest round does not abort:
\begin{align*}
\Pr \left[ \ok \right] = \Pr \left[ Y < w \right].
\end{align*}
Note that $Y$ describes exactly the number of test rounds in which at least one trap measurement outcome is incorrect (by definition of a failed test round). The probability that a given test round fails is therefore upper-bounded by $\pmax$. Let $\hat{Y}_1$ be a random variable following a $(t, \pmax)$-binomial distribution. Since we suppose that the noise is not correlated across rounds, $Y$ is upper-bounded by $\hat{Y}_1$ in the usual stochastic order:
\begin{align*}
\Pr \left[ Y < w \right] \geq &\Pr \left[ \hat{Y}_1 < w \right] = 1 - \Pr \left[ \hat{Y}_1 \geq w \right]
\end{align*}
Further, since $\mathbb{E}\left[ \hat{Y}_1 \right] = t\pmax < w$, applying Lemma~\ref{lemma:hoeffding_binomial} yields:
\begin{align*}
\Pr \left[ \hat{Y}_1 \geq w \right] \leq &\exp \left( -2 \frac{(t \pmax - w)^2}{t} \right)\\
&= \exp \left( -2 (\omega - \pmax)^2 \tau n \right) = \epsilon_{\mathit{rej}}.
\end{align*}

\paragraph{For $\omega < \pmin$.} In that case, we have that $w < t\pmin$. We show that the probability of accepting is upper-bounded by a negligible function. 
Let $\hat{Y}_2$ be a random variable following a $(t, \pmin)$-binomial distribution, $Y$ then is lower-bounded by $\hat{Y}_2$ in the usual stochastic order:
\begin{align*}
\Pr \left[ Y < w \right] \leq \Pr \left[ \hat{Y}_2 < w \right]
\end{align*}
Since $w < t\pmin$, using Lemma~\ref{lemma:hoeffding_binomial} directly and with the same simplifications as above, we get:
\begin{align*}
\Pr \left[ \hat{Y}_2 < w \right] \leq \exp \left( -2 (\pmin - \omega)^2 \tau n \right),
\end{align*}
concluding the proof.
\end{proof}

\begin{theorem}[Local-Correctness of VDQC Protocol on Noisy Devices]
\label{thm:correctness_form}
Assume a Markovian round-dependent model for the noise on Client and Server devices and let $\pmax$ be an upper-bound on the probability that at least one of the trap measurement outcomes in a single test round is incorrect.

If $\pmax < \omega < \frac{1}{k}\cdot \frac{2p-1}{2p-2}$, then the protocol is $\epsilon_{\mathit{cor}}$-locally-correct with exponentially small $\epsilon_{\mathit{cor}} = \epsilon_{\mathit{rej}} + \epsilon_{\mathit{ver}}$, with $\epsilon_{\mathit{rej}}$ from Lemma~\ref{lem:acceptance} and $\epsilon_{\mathit{ver}}$ from Theorem~\ref{thm:verif}.
\end{theorem}

\begin{proof}
We call $\ok$ the event that the Client accepts at the end of the protocol -- if not too many test rounds fail -- and $\Correct$ the event corresponding to a correct output -- if only few of the computation rounds have their output bits flipped.

We are looking to lower-bound the probability of an honest round producing the correct outcome and not aborting:
\begin{align*}
	\Pr \left[ \Correct \land \ok \right] = \Pr \left[ \ok \right] - \Pr \left[ \neg\Correct \land \ok \right].
\end{align*}
As $\pmax < \frac{1}{k}\cdot \frac{2p-1}{2p-2} < 1/2$, from Lemma~\ref{lem:acceptance} we have
\begin{align*}
	\Pr \left[ \ok \right] \geq 1 - \epsilon_{\mathit{rej}}.
\end{align*}
Since $\omega < \frac{1}{k}\cdot \frac{2p-1}{2p-2}$, the parameters of Protocol~\ref{prot:MQ-VBQC} comply with Theorem~\ref{thm:verif}, from which we get that
\begin{align*}
	\Pr \left[ \neg\Correct \land \ok \right] \leq \epsilon_{\mathit{ver}}.
\end{align*}
It follows that
\begin{align*}
	\Pr \left[ \Correct \land \ok \right] \geq 1 - \epsilon_{\mathit{rej}} - \epsilon_{\mathit{ver}},
\end{align*}
which concludes the proof.
\end{proof}

\end{document}